\RequirePackage{amsmath,amssymb}
\documentclass[runningheads]{llncs}
\usepackage{xcolor}
\usepackage{graphicx}
\usepackage{wrapfig}
\usepackage{hyperref}
\usepackage{cleveref}
\usepackage{cite}
\usepackage{orcidlink}
\usepackage{paralist, tabularx}
\usepackage{multirow}
\usepackage{tikz}
\usetikzlibrary{automata, positioning}

\newif\ifcomments
\commentsfalse

\newif\ifshort
\shortfalse

\newif\ifconference
\conferencefalse

\ifcomments
\newcommand{\sharon}[1]{\textcolor{purple}{Sharon: #1 }}
\newcommand{\raz}[1]{\textcolor{teal}{Raz: #1 }}
\newcommand{\neta}[1]{\textcolor{orange}{Neta: #1 }}
\newcommand{\oded}[1]{\textcolor{olive}{Oded: #1 }}
\else
\newcommand{\sharon}[1]{}
\newcommand{\raz}[1]{}
\newcommand{\neta}[1]{}
\newcommand{\oded}[1]{}
\fi
\newcommand{\commentout}[1]{}

\spnewtheorem{constructor}[newcounter]{Constructor}{\bfseries}{\itshape}
\crefname{figure}{fig.}{figs.}
\Crefformat{constructor}{#2Constructor #1#3}
\Crefrangeformat{constructor}{#3Constructor~#1#4--#5#2#6}
\Crefmultiformat{constructor}{Constructor~#2#1#3}{ and~#2#1#3}{, #2#1#3}{ and~#2#1#3}

\crefname{example}{Ex.}{exs.}
\Crefname{example}{Ex.}{Exs.}

\newcommand{\orderformula}[0]{\ell}
\newcommand{\immutorder}[0]{\mathrm{order}}

\newcommand{\formula}[0]{\alpha}
\newcommand{\formulaa}[0]{\beta}
\newcommand{\term}[0]{t}
\newcommand{\domain}[0]{D}
\newcommand{\interp}[0]{I}
\newcommand{\struct}{s}
\newcommand{\structset}{\mathrm{struct}}
\newcommand{\seq}[1]{{\vec #1}}
\newcommand{\assign}{v}
\newcommand{\otherassign}{u}
\DeclareMathOperator{\assignset}{assign}
\newcommand{\concatvar}[0]{{\cdot}}
\newcommand{\concatfunc}[0]{{\cdot}}

\newcommand{\signature}[0]{\Sigma}
\newcommand{\axioms}[0]{\Gamma}
\newcommand{\init}[0]{\iota}
\newcommand{\tr}[0]{\tau}
\newcommand{\specification}[0]{(\axioms,\init,\tr)}
\newcommand{\specsub}[1]{({\axioms_{#1}},\init_{#1},\tr_{#1})}

\newcommand{\Tspec}[0]{\mathcal T}
\newcommand{\semantics}[0]{\mathcal S}
\newcommand{\inv}[0]{\theta}

\newcommand{\globally}[0]{\square}
\newcommand{\eventually}[0]{\lozenge}
\newcommand{\nextop}[0]{\bigcirc}
\newcommand{\until}[0]{\mathcal U}
\newcommand{\temporalformula}[0]{\psi}
\newcommand{\temporalformulaa}[0]{\rho}
\newcommand{\temporalformulaaa}[0]{\chi}

\newcommand{\prop}[0]{\varphi}

\newcommand{\timer}[0]{t}
\newcommand{\timerzeroaxiom}[0]{\zeta}
\newcommand{\timerzerotr}[0]{\eta}
\newcommand{\sorts}[0]{\mathrm{sorts}}
\newcommand{\sort}[0]{\sigma}
\newcommand{\timename}[0]{\mathrm{time}}
\newcommand{\timesort}[0]{\sort_\timename}
\newcommand{\timesemantics}[0]{\semantics_\timename}
\newcommand{\timersystem}[0]{\Tspec_{\prop}(\timesemantics)}
\newcommand{\nat}{\mathbb{N}}
\newcommand{\natinf}{\nat \cup \{ \infty \}}
\newcommand{\sub}[1]{A({#1})}

\newcommand{\rankname}[0]{\mathrm{Rk}}
\newcommand{\reduced}[0]{\tau_{\scriptscriptstyle \mathord{>}}}
\newcommand{\reducedsuper}[1]{\tau_{\scriptscriptstyle \mathord{>}}^{#1}}
\newcommand{\conserved}[0]{\tau_{\scriptscriptstyle \mathord{\geq}}}
\newcommand{\conservedsuper}[1]{\tau_{\geq}^{#1}}
\newcommand{\minformula}[0]{\kappa_{\textnormal{min}}}
\newcommand{\minsuper}[1]{\kappa_{\textnormal{min}}^{#1}}
\newcommand{\condition}[0]{\ensuremath{\mathcal C}}
\newcommand{\conditionsuper}[1]{\condition^{#1}}
\newcommand{\ranktuple}[0]{(\reduced,\conserved,\minformula)}
\newcommand{\ranktuplesup}[1]{(\reducedsuper{#1},\conservedsuper{#1},\minsuper{#1})}
\newcommand{\insup}[0]{\circ}
\newcommand{\bijec}[0]{\sigma}

\newcommand{\pair}[0]{(\struct,\assign)}
\newcommand{\pairlow}[0]{(\struct',\assign')}
\newcommand{\pairhigh}[0]{\pair}
\newcommand{\twopair}[0]{\pairhigh,\pairlow}

\begin{document}

\title{Verifying First-Order Temporal Properties of Infinite-State Systems via Timers and Rankings}
\author{
Raz Lotan
\inst{1,2}
\orcidlink{0009-0008-5883-5082} 
\and
Neta Elad
\inst{1}
\orcidlink{0000-0002-5503-5791} 
\and
Oded Padon
\inst{3}
\orcidlink{0009-0006-4209-1635} 
\and 
Sharon Shoham
\inst{1}
\orcidlink{0000-0002-7226-3526}
}

\institute{
Tel Aviv University, Tel Aviv, Israel
\and
Certora, Tel Aviv, Israel
\and
Weizmann Institute of Science, Rehovot, Israel\\
\email{lotanraz@gmail.com
}
}

\authorrunning{R. Lotan, N. Elad, O. Padon, S. Shoham}

\titlerunning{Verifying FO-LTL Properties via Timers and Rankings}

\maketitle

\begin{abstract} 
We present a unified deductive verification framework for first-order temporal properties based on well-founded rankings, where verification conditions are discharged using SMT solvers.
To that end, we introduce a novel reduction from verification of arbitrary temporal properties to verification of termination. 
Our reduction augments the system with prophecy timer variables that predict the number of steps along a trace until the next time certain temporal formulas, including the negated property, hold. In contrast to standard tableaux-based reductions, which reduce the problem to fair termination, our reduction does not introduce fairness assumptions.
To verify termination of the augmented system, we follow the traditional approach of assigning each state a rank from a well-founded set and showing that the rank decreases in every transition.
We leverage the recently proposed formalism of implicit rankings to express and automatically verify the decrease of rank using SMT solvers, even when the rank is not expressible in first-order logic.
We extend implicit rankings from finite to infinite domains, enabling verification of more general systems and making them applicable to the augmented systems generated by our reduction, which allows us to exploit the decrease of timers in termination proofs.
We evaluate our technique on a range of temporal verification tasks from previous works, giving simple, intuitive proofs for them within our framework. 
\end{abstract}

\section{Introduction}
\label{section:Introduction}
Several recent works~\cite{liveness_to_safety,prophecy,towards_liveness_proofs,ImplicitRankings,LVR} have dealt with verifying temporal properties of transition systems specified in first-order logic.
Temporal properties of such systems are naturally specified in First-Order Linear Temporal Logic (FO-LTL).
FO-LTL extends standard first-order logic with temporal operators that reason about the behavior of a transition system over time, 
and allows to specify both safety properties and liveness properties with various fairness assumptions.

Existing approaches for the verification of liveness properties of transition systems given in first-order logic can roughly be divided to those employing a liveness-to-safety reduction or those that use well-founded rankings.
The liveness-to-safety reduction of~\cite{liveness_to_safety} can target any FO-LTL property and is quite powerful when augmented with temporal prophecy~\cite{prophecy,power_of_temporal_prophecy}.
However, verifying the safety of the resulting system is complicated since it requires reasoning about a monitored copy of the system that tracks repeated states modulo a dynamic finite abstraction.
In contrast, well-founded rankings operate at the level of the original system
and are therefore more intuitive to use,
but current works in this setting~\cite{towards_liveness_proofs,ImplicitRankings,LVR} 
only target specific liveness properties and are limited in the kinds of ranking arguments they can capture.

Our goal is to design a  verification framework for first-order transition systems that is both intuitive for users and applicable to any temporal property. 
To this end, we combine two key ideas. 
First, to obtain a unified approach, we base our framework on a novel reduction from the verification of FO-LTL properties to verification of termination. 
Our reduction is sound and complete and produces a transition system that is easy to reason about since it is an augmentation of the original transition system with prophecy variables that track violations of the property.
Then, verifying termination of the resulting transition system is equivalent to showing that there is no violating trace of the original system.
Second, for verifying termination, we adopt and extend the notion of implicit rankings from~\cite{ImplicitRankings}, allowing the user to encode a range of well-founded rankings that may involve the prophecy variables introduced by the reduction.
Implicit rankings are provided by the user and verified automatically by an SMT solver. 

Our reduction to termination takes a transition system $T$ and a temporal property $\prop$ and reduces the verification of $T\models\prop$ to verifying termination of an augmented transition system obtained by a synchronous composition of $T$ with $T_{\neg\prop}$, the \emph{timer transition system} of $\neg \prop$. 
The infinite traces of $T_{\neg\prop}$ are exactly all the infinite traces that satisfy $\neg\prop$.
The timer transition system resembles standard tableaux constructions, except that it does not have fairness assumptions. 
As a result, we get a reduction to verification of termination instead of fair termination, which is easier to handle.
To construct $T_{\neg \prop}$ we introduce for each subformula 
$\temporalformula$ of $\neg \prop$ a timer variable $\timer_\temporalformula$ that takes values in $\natinf$, and predicts the number of transitions until $\temporalformula$ is next satisfied, or $\infty$ if there is no such number. This behavior is enforced by the transitions of $T_{\neg \prop}$.
The initial states of $T_{\neg \prop}$ are then defined as those where $\timer_{\neg\prop} = 0$ which means that all of their outgoing infinite traces satisfy $\neg\prop$. 
It follows that the composition $ T\times T_{\neg\prop}$ is an augmentation of $T$ with timers that contains all infinite traces of $ T$ that satisfy $\neg \prop$. Thus, $ T\models\prop$ holds if and only if 
the augmented transition system obtained by the composition has no infinite traces, that is, terminates.

To show termination of the transition systems produced by the reduction we employ the method of implicit rankings~\cite{ImplicitRankings}. 
Implicit rankings are first-order formulas that capture the decrease of some ranking function between two states, without defining the function explicitly. 
Implicit rankings are useful for reasoning about ranking functions that are not easily expressible in first-order logic, such as functions that involve unbounded set cardinalities or summations.
The constructors of implicit rankings defined in~\cite{ImplicitRankings} provide an expressive language for well-founded rankings that can be reasoned about in first-order logic, using SMT solvers, but these are defined only for systems that operate over finite, albeit unbounded, domains. 
In particular, they cannot leverage timers in rankings, even though their domain is well-founded.
To mitigate this, we extend the formalism of implicit rankings 
to be able to soundly reason over systems with infinite domains.
We do so by introducing an explicit soundness condition that allows us, for example, to require that some dynamically updated set is finite in all reachable states or that some relation is well-founded.
We further generalize the constructors of~\cite{ImplicitRankings} to this setting.
In particular, this allows us to use the decrease of timers in the ranking used to prove termination of the augmented system produced by our reduction.
Importantly, the use of timers in the termination proof relies on their intuitive meaning, and does not require understanding the construction of the augmented system.

We implement our reduction and the constructions of implicit rankings in a deductive verification tool that takes as input a transition system specification and a first-order temporal property, as well as an implicit ranking (given by a composition of constructors) and supporting invariants, which may refer to the timer variables. Our tool uses these ingredients to generate verification conditions that capture the decrease of rank encoded by the implicit rankings in every reachable transition of the system augmented with timers.
It then verifies that the system satisfies the property by using an SMT solver to automatically discharge the verification conditions. 
We evaluate our tool by verifying a set of challenging examples from previous works which include liveness properties of distributed protocols and termination of programs. 
We give simple intuitive proofs for these examples within our framework.
\paragraph{Contributions.}
\begin{itemize}
    \item We introduce a reduction from the verification of arbitrary FO-LTL properties to verification of termination that augments the transition system with prophecy timer variables.
    \item We generalize the definition of implicit rankings and their constructors from recent work to settings with infinite domains.
    \item We combine the reduction to termination with implicit rankings to obtain a deductive verification framework where a user constructs an implicit ranking for the augmented system, and the decrease of rank in every step of the augmented system is verified automatically by an SMT solver.
    \item We implement and evaluate our approach on a set of transition systems and their properties from the literature. 
\end{itemize}
The rest of the paper is organized as follows:
\Cref{section:motivatingExample} outlines a motivating example;
\Cref{section:prelims} gives background on first-order logic and FO-LTL;
\Cref{section:timerReduction} introduces the reduction from FO-LTL to termination;
\Cref{section:rankings} presents the generalized definition of implicit rankings, the proof rule that uses them to prove termination and constructors for implicit rankings;
\Cref{section:evaluation} details our implementation and evaluation;
\Cref{section:relatedWork} discusses related work and concludes the paper.
\ifconference
\else
We defer all proofs to \Cref{appendix:proofs}.
\fi
\newpage
\section{Running Example}
\label{section:motivatingExample}

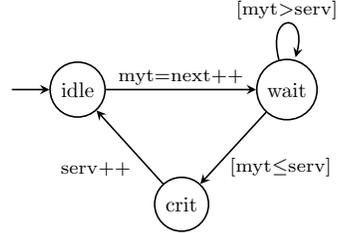
\begin{wrapfigure}{r}{0.35\textwidth}
\vspace{-0.5cm}
\centering
\begin{tikzpicture}[->, >=stealth, auto, semithick, node distance=2.5cm
]
\node[state, scale=0.8] (idle) {idle};
\node[state, scale=0.8] (wait) [right=2cm of idle] {wait};
\node[state, scale=0.8] (crit) [below=0.75cm of wait, xshift=-1.75cm] {crit};
\path (idle) edge node[yshift=-2pt] {$\scriptstyle
\mathrm{myt} = \mathrm{next++}$} (wait)
(wait) edge[loop above] node[yshift=-2pt]  {$\scriptstyle [\mathrm{myt} > \mathrm{serv}]$}  (wait)
(wait) edge node[xshift=-5pt] {$\scriptstyle [\mathrm{myt} \leq \mathrm{serv}]$} (crit)
(crit) edge node[xshift=4pt,yshift=-2pt]  {$\scriptstyle\mathrm{serv}++$} (idle);
\draw[->] ([xshift=-0.5cm]idle.west) -- (idle.west);
\end{tikzpicture}
\caption{Ticket Protocol}
\vspace{-0.3cm}
\label{figure:Ticket}
\end{wrapfigure}

To illustrate our approach we consider the ticket protocol for ensuring mutual exclusion, a variant of Lamport's bakery algorithm~\cite{lamportBakery}.
In the ticket protocol, access to the critical section is governed by natural-number tickets.
The protocol maintains the next available ticket and the currently serviced ticket, both initialized to zero.
Threads start at an idle state and may move to a waiting state by acquiring the next available ticket and in turn increasing it.
Once in the waiting state, a thread must wait for its ticket to be the currently serviced ticket  to enter the critical section. 
Immediately after a thread enters the critical section, it transitions back to the idle state while increasing the currently serviced ticket.
The control-flow graph of each thread is presented in \Cref{figure:Ticket}.

The temporal property we verify is that under fair scheduling, every thread that waits to enter the critical section, eventually does so.
We formalize this property in FO-LTL by
$ 
\prop = 
(\forall x \ \globally\eventually \mathrm{scheduled}(x))
\to 
(\forall x \ \globally (\mathrm{waiting}(x)\to \eventually \mathrm{critical}(x)))
$ where $\globally$ stands for ``globally'' and $\eventually$ stands for ``eventually''.
To verify the property, we need to show that no infinite trace of the system satisfies the negation of $\prop$.
We do so by considering an augmented transition system that restricts the infinite traces of the original system to traces that satisfy the skolemized, negated property
$
\mathrm{sk}(\neg \prop) = 
\forall x. \globally\eventually \mathrm{scheduled}(x) \wedge \eventually (\mathrm{waiting}(x_0)\wedge  \globally \neg \mathrm{critical}(x_0))$.
Verification is then reduced to showing that the augmented transition system has no infinite traces, i.e., that it terminates.

The augmented transition system is constructed automatically by composing the original one with a timer transition system for $
\mathrm{sk}(\neg \prop)$. 
The general construction is given in \Cref{subsection:constructors}. Here, we only illustrate the gist of the construction. 
The augmented transition system has a timer variable $\timer_{\mathrm{starved}}\in \natinf$ and a  timer function $\timer_{\mathrm{sched}}: \mathrm{Thread}\to\natinf$.
The variable $\timer_{\mathrm{starved}}$ serves as a prophecy variable that predicts the number of steps until $\mathrm{waiting}(x_0)\wedge  \globally \neg \mathrm{critical}(x_0)$
will next hold. When $\timer_{\mathrm{starved}}$ reaches $0$ it means that $x_0$ is forever locked out of the critical section.
The function $\timer_{\mathrm{sched}}(x)$ predicts the number of steps until the thread $x$ is next scheduled.
The construction ensures that the timers are updated according to this meaning: if their value is positive, it decreases in every step; if it is infinite, it remains infinite; and if it is $0$, it is ensured that the formula holds.
Because we want to guarantee that $\mathrm{sk}(\neg \prop)$,
the augmented transition system initializes $\timer_{\mathrm{starved}}$ to some natural number, it then decreases in every step of the system, which, due to the nature of the natural numbers, means it becomes 0 after finitely many steps. 
Additionally, to guarantee fair scheduling, for every thread $x$, the construction ensures that $\timer_{\mathrm{sched}}(x)$ is equal to 0 infinitely often. It does so by initializing it to some natural number and resetting it to an arbitrary positive number whenever it reaches 0.
Altogether this guarantees that \emph{any} infinite trace of the augmented system satisfies $\mathrm{sk}(\neg \prop)$ (without any fairness assumptions).
It then remains to prove that the augmented system terminates.

In the remainder of the section we illustrate how timers are leveraged in the termination proof of the augmented system. We emphasize that to incorporate the timers in the proof, the user may rely on their intuitive meaning and need not worry about the construction of the augmented system.
Our termination proof is based on a ranking function that maps system states to a rank from some well-founded set that decreases in every transition. However, instead of encoding the function explicitly, we use an extended notion of implicit rankings~\cite{ImplicitRankings} which encode decrease in the rank between two states, and can be naturally defined using the constructors we provide (see \Cref{subsection:constructors}).
 Due to space constraints, we only present the underlying ranking function. Encoding it as an implicit ranking using our constructors is rather straightforward, given our extension that allows to incorporate timers in implicit rankings despite their infinite domain.

We construct a lexicographic ranking function by considering the possible transitions of the (augmented) system. 
The first component of the rank is $\timer_\mathrm{starved}$. This timer repeatedly decreases until it reaches 0, from which point on it remains 0, and $x_0$ is starved out of the critical section.
The next components of the rank establish that the remainder of the trace from that point on is finite. 
To that end, we observe that since $x_0$ is in the waiting state indefinitely, its ticket value
$\mathrm{myt}(x_0)$ never changes and therefore the difference $\mathrm{myt}(x_0) - \mathrm{serv}$ never increases.
When a thread leaves the critical state $\mathrm{serv}$ increases  and the difference decreases.
Therefore, we use the difference as the next component in the rank.
Next, when no thread is in the critical section we wait for one to enter. 
Therefore, the third component in the rank captures whether there is a thread in the critical section. 
Finally, the scheduled thread may neither be in nor enter the critical section.
To ensure decrease of rank in this case, we leverage the timers and introduce a fourth component of the rank that decreases according to $\timer_{\mathrm{sched}}(x)$ for the thread $x$ that holds the currently serviced ticket and is the next one to enter the critical section.
When the resulting rank is encoded with an implicit ranking, the decrease of rank in every  (reachable) step of the augmented transition system can be verified automatically by an SMT solver, ensuring termination.

\section{Preliminaries}
\label{section:prelims}

\paragraph{First-Order Logic.}
We use many-sorted first-order logic with equality. A signature $\signature$ consists of a set of sorts, denoted $\mathrm{sorts}(\signature)$, sorted relation, constant and function symbols.
Terms are constant symbols, variables 
or function applications.
We use boldface to denote a sequence of variables $\seq x$ or terms $\seq \term$.
Atomic formulas are $\term_1=\term_2$
or relation applications $R(\seq t)$.
Non-atomic formulas are built using the  connectives $\neg,\wedge,\vee,\to$ and sorted quantifiers $\forall,\exists$.
For a formula $\formula$ (or term), we write $\formula(\seq x)$ to denote that its free variables are contained in $\seq x$, and for a sequence of terms $\seq \term$ we write $\formula(\seq \term)$ for the simultaneous substitution $\alpha[\seq x \mapsto \seq \term]$.

First-order formulas are evaluated over pairs of structures and assignments. 
A $\signature$-structure is a pair $\struct =(\domain,\interp)$ where $\domain$ maps every sort $\sort$ to its domain $\domain(\sort)$, which is a non-empty set, 
and $\interp$ is an interpretation function for all symbols in $\signature$. 
We denote by $\structset(\signature,\domain)$ the set of all structures over $\domain$ and by $\structset(\signature)$ the class of all $\Sigma$-structures. 
A sorted assignment $\assign$ from 
sorted variables $\seq x$ to $\domain$ is a function that maps each variable $x_i$ of sort $\sort$ in $\seq x$ 
to an element of $\domain(\sort)$.
We denote the set of all  assignments from $\seq x$ to $\domain$ by $\assignset(\seq x,\domain)$.
For two assignments $\otherassign,\assign$ over disjoint sequences of variables $\seq x, \seq y$, we denote by $\otherassign\concatfunc\assign$ the assignment to $\seq x\concatvar \seq y$ defined by $\otherassign\concatfunc\assign(x_i) = \otherassign(x_i)$ and $\otherassign\concatfunc\assign(y_j) = \assign(y_j)$ for any $i,j$.
For a formula $\formula$ and a pair $(\struct,\assign)$ 
of a structure $\struct$ and an assignment $\assign$ to the free variables of $\alpha$ we write $(\struct,\assign)\models \formula$ to denote that $(\struct,\assign)$ satisfies $\formula$.

For a signature $\signature$ we denote by $\signature'$ a disjoint copy of $\signature$ defined by $\signature' = \{a' \mid a\in \signature\}$.
For a sequence of variables $\seq x = (x_i)_{i=1}^m$, we denote by $\seq x'$ the sequence $(x'_i)_{i=1}^m$.
For a formula $\formula$ (or term $t$) over $\signature$, we denote by $\formula'$ the formula over $\signature'$ obtained by substituting each symbol $a\in \signature$ with $a'\in \signature'$.
With abuse of notation we sometimes consider a $\signature$-structure $\struct=(\domain,\interp)$ as a $\signature'$-structure, or vice versa, by setting $\interp(a')=\interp(a)$.
Similarly, we sometimes consider an assignment to $\seq x$ as an assignment to $\seq x'$, or vice versa, by setting $\assign(x'_i)=\assign(x_i)$.
For a formula $\tr$ over $\signature \uplus \signature'$ with free variables $\seq x, \seq x'$ and structures $\struct,\struct'$ over a shared domain, we use the notation $(\struct,\assign),(\struct',\assign')\models \tr$ to denote that $\tr$ is satisfied by the structure-assignment pair where the domain is the shared domain of $\struct,\struct'$ and where  $\signature,\seq x$ are interpreted according to $(\struct,\assign)$ and $\signature',\seq x'$
are interpreted as in $(\struct',\assign')$. 
If $\tr$ is closed we omit the assignments and write $\struct,\struct'\models\tr$. 

\paragraph{Specifying Transition Systems in First-Order Logic.} 
A transition system over signature $\Sigma$, denoted $\Tspec(\semantics)$, is given by a pair of a first-order specification $\Tspec$ and an intended semantics $\semantics$. 
The first-order specification is given by $\Tspec = \specification$ where 
$\axioms$  is a closed formula over $\signature$ that acts as a background theory, 
$\init$ is a closed formula over $\signature$ that defines the initial states of the system 
and $\tr$ is a closed formula over $\signature\uplus\signature'$ that defines the possible transitions of the system. 
The intended semantics $\semantics$ is a class of $\signature$-structures.
States of $\Tspec(\semantics)$ are given by the $\signature$-structures in $\semantics$ that satisfy $\axioms$.
A trace of $\Tspec(\semantics)$ is given by a sequence of states $\pi=(\struct_i)_{i=0}^n$ where $n\in \nat\cup\{\infty\}$ such that all states in the sequence have the same domain, $\struct_0 \models \init$ and for every $i\leq n$ we have $\struct_i,\struct_{i+1} \models \tr$. 
A state $\struct$ is reachable if there is a finite trace $\pi = (\struct_i)_{i=0}^n$ for $n\in\nat$ such that $\struct_n = \struct$. A pair of states $\struct,\struct'$ is a reachable transition if $\struct$ is reachable and $\struct,\struct' \models \tr$.
We say that $\Tspec(\semantics)$ \emph{terminates} if it has no \emph{infinite} traces.

\paragraph{First-Order Linear Temporal Logic.}
We use FO-LTL to specify temporal properties of transition systems.
Given a first-order signature $\signature$,
FO-LTL formulas are defined similarly to  first-order formulas with the addition of temporal operators.
For simplicity we only consider the temporal operators `globally' $\globally$ and `eventually' $\eventually$, our approach can be easily extended to the operators 'next' $\nextop$ and 'until' $\until$.
The semantics of an FO-LTL formula $\temporalformula$ is defined over a pair $(\pi,\assign)$, where $\pi$ is an \emph{infinite} sequence of $\signature$-structures $\pi = (\struct_i)_{i=0}^\infty$ over a shared domain, and $\assign$ is an assignment to the free variables of $\temporalformula$.
For 
$\pi = (\struct_i)_{i=0}^\infty$ 
and $k\in\nat$ we denote 
$\pi^k = (\struct_{i+k})_{i=0}^\infty$.
The semantics of FO-LTL for $\temporalformula$ and $(\pi,\assign)$ as above is such that if $\temporalformula$ is atomic we have 
$\pi,\assign\models\temporalformula$ if and only if $\struct_0,\assign\models\temporalformula$ in the first-order semantics, 
$\pi,\assign\models\globally\temporalformula$ if and only if $\pi^k,\assign\models\temporalformula$ for every $k\in\nat$,
$\pi,\assign\models\eventually\temporalformula$ if and only if $\pi^k,\assign\models\temporalformula$ for some $k\in\nat$, and the first-order operators are interpreted in the usual way.
A temporal property $\prop$ is given by a closed FO-LTL formula.
A transition system $\Tspec(\semantics)$ satisfies a temporal property $\prop$, denoted $\Tspec(\semantics) \models \prop$, if all of its infinite traces satisfy $\prop$.

\section{Timers}
\label{section:timerReduction}
In this section we present the reduction from verification of FO-LTL properties to verification of termination.
First, in \Cref{subsection:timerSystem}, we present the timer transition system of a temporal property $\prop$, and prove two characteristic facts about it.
Then, in \Cref{subsection:reduction}, we present the reduction, which uses the timer transition system of the negated temporal property to be verified.

\subsection{The Timer Transition System}
\label{subsection:timerSystem}

The timer transition system of a temporal property $\prop$
is a transition system constructed such that (i)~all of its infinite traces satisfy $\prop$ and (ii)~it contains all infinite traces that satisfy $\prop$. 
We construct the timer transition system by considering a new sort, denoted $\timesort$, whose intended semantics
are of the extended natural numbers $\natinf$, 
and adding for every subformula $\temporalformula$ a function symbol $\timer_\temporalformula$ called the \emph{timer} of $\temporalformula$.

Formally, for an FO-LTL property $\prop$ over signature $\signature$ we denote by 
$\sub{\prop}$ the set of formulas that contains all subformulas of $\prop$, and for every subformula $\globally\temporalformula$, additionally contains $\neg\temporalformula$.
The timer transition system 
$\timersystem$ is defined over signature
$\signature_\prop$ with sorts $\sorts(\signature_\prop)= \sorts(\signature) \cup \{ \timesort\}$ and signature $\signature_\prop = \signature \cup 
\{ \timer_{\temporalformula} \mid \temporalformula\in \sub{\prop} \}
\cup 
\{ \leq, 0, \infty, -1\}$,
where $\timer_\temporalformula$ for $\temporalformula(\seq x) \in \sub{\prop}$ is a function symbol
from the sorts of $\seq x$ to $\timesort$;
$\leq$ is a binary relation symbol on $\timesort$;  $0$ and $\infty$ are constant symbols of sort $\timesort$; and 
$-1$ is a function symbol from $\timesort$ to $\timesort$.
The timer intended semantics $\timesemantics$ are given by the class of structures for $\signature_\prop$ where the domain of $\timesort$ is $\natinf$, with the standard interpretations for 
${\leq}, {0}, {\infty}$ and $-1$.
The axioms and transition formula of $\Tspec_\prop$ are defined such that for any $\temporalformula\in\sub{\prop}$, $\timer_\temporalformula(\seq x)$ captures the number of transitions required until $\temporalformula(\seq x)$ is satisfied.
Finally, the initial states formula ensures that $\timer_\prop = 0$ holds in all initial states, which will imply that $\prop$ is satisfied in every trace.

\begin{definition}
\label{definition:timerSystem}
Let $\prop$ be an FO-LTL property over signature $\signature$, and define $\sub\prop, \allowbreak\signature_\prop$ and $\timesemantics$ as above.
The \emph{timer transition system} of $\prop$ is a transition system $\timersystem$ over $\signature_\prop$,
where $\Tspec_\prop=(\axioms_\prop,\init_\prop,\tr_\prop)$ is defined by: 
\[
\begin{array}{rcll}
\axioms_\prop &=& 
\displaystyle
\bigwedge_{\temporalformula(\seq x) 
\in \sub{\prop}} &
\forall \seq x \
\left((\timer_\temporalformula(\seq x) = 0) \leftrightarrow \timerzeroaxiom_\temporalformula(\seq x)\right)
\qquad\qquad 
\init_\prop = (\timer_\prop = 0)\\
\tr_\prop &=& 
\displaystyle
\bigwedge_{\temporalformula(\seq x) \in \sub{\prop}} &
\forall \seq x
\left(
\begin{array}{ll}
     &
    (0 < \timer_\temporalformula(\seq x) <\infty \to \timer'_\temporalformula(\seq x)=\timer_\temporalformula(\seq x) -1) \ \wedge \\
    & 
    (\timer_\temporalformula(\seq x) = \infty \to \timer'_\temporalformula(\seq x) = \infty) \ \wedge \\
    & 
    (\timer_\temporalformula(\seq x) = 0) \leftrightarrow \timerzerotr_\temporalformula(\seq x)  
\end{array}
\right)
\end{array}
\]
    where $\timerzeroaxiom_\temporalformula$ and $\timerzerotr_\temporalformula$ are defined according to \Cref{table:timerSpecification}. If $\timerzerotr_\temporalformula$ is undefined, the conjunct that refers to it is omitted from $\tr_\prop$.
\begin{table}[t]
    \centering
    \caption{Definitions of 
$\timerzeroaxiom_\temporalformula,\timerzerotr_\temporalformula$
    for an FO-LTL formula $\temporalformula(\seq x)$
    }
    \scalebox{0.9}{
    \bgroup
    \def\arraystretch{1.5}
    \setlength\tabcolsep{5pt}
    \begin{tabular}{| c | c | c |}
\hline
$\temporalformula(\seq x)$ & $\timerzeroaxiom_\temporalformula$ & $\timerzerotr_\temporalformula$ \\
\hline 
$R(\seq x)$
&
$R(\seq x)$
&
\\
$\neg \temporalformulaa$
&
$ \timer_\temporalformulaa(\seq x) \neq 0$
&
\\
$\temporalformulaa \circ \temporalformulaaa$
for $\circ\in \{\wedge,\vee,\to\}$ &
$(\timer_{\temporalformulaa}(\seq x) = 0) \circ (\timer_{\temporalformulaaa}(\seq x) = 0)$
&
\\
$Q y \ \temporalformulaa$ for $Q\in \{\forall,\exists\}$
&
$Q y\ (\timer_{\temporalformulaa}(\seq x\concatvar y) = 0)$
&
\\
$\eventually \temporalformulaa$
& 
$\timer_\temporalformulaa(\seq x) < \infty$
&
$\timer_\temporalformulaa(\seq x) = 0 \vee t'_{\eventually \temporalformulaa}(\seq x)=0$
\\
$\globally \temporalformulaa$
& 
$\timer_{\neg\temporalformulaa}(\seq x)=\infty$
&
$\timer_\temporalformulaa(\seq x) = 0 \wedge t'_{\globally \temporalformulaa}(\seq x)=0$
\\

\hline
    \end{tabular}
    \egroup
    }
    \label{table:timerSpecification}
\end{table}
\end{definition}

\begin{example}
    Consider $\prop = \forall x. \globally\eventually \mathrm{sched}(x)$.
    The timers in $\Tspec_\prop$ are
    $\timer_{\forall x. \globally\eventually \mathrm{sched}(x)}, \allowbreak
    \timer_{\globally\eventually \mathrm{sched}(x)}(x), 
    \timer_{\neg\eventually \mathrm{sched}(x)}(x),
    \timer_{\eventually \mathrm{sched}(x)}(x), 
    \timer_{\mathrm{sched}(x)}(x)$.
    In initial states of $\Tspec_\prop$ we have $\timer_\varphi = 0$, and so from $\Gamma_\prop$ we have $\forall x. \timer_{\globally\eventually \mathrm{sched}(x)}(x) = 0$ in all initial states. 
    From $\tr_\prop$ it follows that $\forall x. \timer_{\eventually \mathrm{sched}(x)}(x)=0$ in all reachable states.
    From the axiom generated for $\eventually$, for any thread $x$, $\timer_{\mathrm{sched}(x)}(x)<\infty$ in any reachable state.
    If $\timer_{\mathrm{sched}(x)}(x)$, from decrease of timers in $\tr_\prop$ it follows that this timer decreases to 0 in finitely many steps.  
    Thus, for any thread $x$ we have infinitely often $\timer_{\mathrm{sched}(x)}(x) = 0$ in any trace, and so fair scheduling is guaranteed.
    That is, every infinite trace of $\Tspec_\prop$ satisfies $\prop$.
\end{example}

\begin{remark}
    The construction of the timer transition system resembles tableaux constructions~\cite{automata_approach,prophecy}. 
    Instead of using sets of temporal formulas as states, we augment states with timers. 
    In the tableau construction, fairness constraints are used to filter out undesirable traces, 
    such as infinite traces where all states are marked $\eventually p$ on the basis of some future satisfication of $p$, which never arrives.
    Instead, in our construction, we rely on the intended semantics of timers to prevent such cases.
    Since such traces induce infinitely decreasing timer values, the well-foundedness of the natural numbers ensures that they do not exist.    
\end{remark}
\ifshort
\else
\begin{remark}
    To handle any FO-LTL formula, we would consider additionally the temporal operators $\nextop$ and $\until$.
    To handle a temporal subformula $\nextop\temporalformulaa$, we would conjoin to $\tr_\prop$ the assertion $\timer_{\nextop\temporalformulaa} = 0 \leftrightarrow \timer_\temporalformulaa' = 0$.
    To handle a temporal subformula $\temporalformulaaa \until \temporalformulaa$, we would conjoin to $\Gamma_\prop$ the assertion $\timer_{\temporalformulaaa \until \temporalformulaa}=0 \to \timer_\temporalformulaa < \infty$, and conjoin to $\tr_\prop$ the assertion 
    $\timer_{\temporalformulaaa \until \temporalformulaa} = 0 \leftrightarrow \timer_\temporalformulaa = 0 \vee (\timer_\temporalformulaaa = 0 \wedge \timer_{\temporalformulaaa \until \temporalformulaa}' = 0)$.
\end{remark}
\fi
Our first lemma states that for any
$\temporalformula\in\sub{\prop}$
timers evaluated to $0$ capture
the satisfaction of the corresponding temporal formula.
In $\init_\prop$ we require that $\timer_\prop = 0$, so it follows that every infinite trace of
$\timersystem$ satisfies $\prop$.
\begin{lemma}
\label{lemma:timerTemporalEquivalence}
Let $\hat\pi = (\hat\struct_i)_{i=0}^\infty$ be a trace of $\timersystem$ with domain $\domain$, $\temporalformula(\vec x)\in \sub{\prop}$,  $\assign\in\assignset(\vec x,\domain)$ and $i\in\nat$. Then $\hat{\struct}_i,\assign \models (\timer_{\temporalformula}(\vec x) = 0)$ if and only if $\hat{\pi}^i,\assign \models \temporalformula(\vec x)$.
\end{lemma}
The following lemma states that any sequence of states that satisfies $\prop$ can be augmented to produce a trace of $\timersystem$. 
This essentially means that $\timersystem$ exhibits all the ways in which $\prop$ can be satisfied.
The proof of the lemma works by showing that we can always augment traces with the ``natural definition'' of timers, which is to evaluate $\timer_\temporalformula(\seq x)$ to the minimal number of transitions that should be taken until the suffix of the  sequence satisfies $\temporalformula$, or $\infty$ if none exists.

\begin{lemma}
\label{lemma:traceExtension}
Given a sequence of $\signature$-structures $\pi = (\struct_i)_{i=0}^\infty$ over a shared domain such that $\pi \models \prop$, there exists a trace $\hat\pi = (\hat\struct_i)_{i=0}^\infty$ of $\timersystem$ such that $\hat{\pi}|_\signature = \pi$.
\end{lemma}

\ifshort\else
\begin{remark}
While the proof of \Cref{lemma:traceExtension} shows that we can always augment traces with the natural definition of timers according to satisfaction of temporal formulas in the trace, \Cref{lemma:timerTemporalEquivalence} only states that traces of $\Tspec_{\prop}(\timesemantics)$ enforce the natural definition for timers evaluated to $0$. 
However, we can in fact show that the same holds for any timer value $k\in\natinf$: for a trace $\hat\pi = (\hat\struct_i)_{i=0}^\infty$ and assignment $\assign$, for $k\in \nat$ we have $\pi^i,\assign \models (t_\temporalformula(\seq x) = k)$ if and only if $ \pi^{i+k},\assign \models \temporalformula$ and $\pi^i,\assign \models (t_\temporalformula(\seq x) = \infty)$ if and only if  $\pi^{i+j},\assign  \nvDash \temporalformula$ for all $j\in \nat$. 
This holds because $\Tspec_{\prop}(\timesemantics)$ requires that timers evaluated to $\infty$ remain at $\infty$ after transitions.
Because we do not use this requirement in the proof of \Cref{lemma:timerTemporalEquivalence}, and removing it only adds traces to $\Tspec_{\prop}(\timesemantics)$ and hence does not affect \Cref{lemma:traceExtension}, we can omit this requirement from the definition of $\Tspec_{\prop}(\timesemantics)$ without affecting soundness and completeness of the reduction presented in the next subsection.
\end{remark}
\fi

\subsection{The Timer Reduction}
\label{subsection:reduction}

The reduction from satisfaction of FO-LTL properties to termination essentially composes $\Tspec(\semantics)$ with the timer transition system of the negated property $\Tspec_{\neg \prop}(\timesemantics)$. Formally, for two transition systems over signatures $\signature_1,\signature_2$ with specifications $\Tspec_1=\specsub{1}$, $\Tspec_2=\specsub{2}$ and intended semantics $\semantics_1,\semantics_2$ respectively, 
we define their \emph{product transition system} over $\signature_1\cup\signature_2$ where $\Sigma_1\cup\Sigma_2$ is not necessarily a disjoint union.
The product transition system is denoted $(\Tspec_1\times \Tspec_2)(\semantics_1\times\semantics_2)$ and defined such that $\Tspec_1\times \Tspec_2=(\axioms_1\wedge\axioms_2,\init_1\wedge\init_2,\tr_1\wedge\tr_2)$,  
and $\semantics_1\times \semantics_2$ contains any $(\Sigma_1\cup\Sigma_2)$-structure $\struct$ such that $\struct|_{\Sigma_1}\in\semantics_1$ and $\struct|_{\Sigma_2}
\in\semantics_2$.
In our case, $\Tspec(\semantics)$ and $\Tspec_{\neg \prop}(\timesemantics)$ are defined over signatures $\signature, \signature_\prop$ respectively such that $\signature \subseteq \signature_\prop$. 
Hence, the product transition system is defined over $\signature_\prop$ and can be understood as an augmentation of $\Tspec(\semantics)$ with the timer variables and semantics of $\Tspec_{\neg \prop}(\timesemantics)$. 
The product system $\Tspec(\semantics) \times \Tspec_{\neg \prop}(\timesemantics)$ contains those traces of $\Tspec(\semantics)$ that satisfy $\neg \prop$, it follows that $\Tspec(\semantics)$ satisfies the property if and only if the product system terminates.
The following theorem summarizes the soundness and completeness of the reduction to termination.

\begin{theorem}
\label{theorem:reductionSoundComplete}
For a transition system $\Tspec(\semantics)$ and an FO-LTL property $\prop$ we have
$\Tspec(\semantics)\models\prop$ if and only if $\Tspec(\semantics) \times \Tspec_{\neg \prop}(\timesemantics)$ terminates.
\end{theorem}

\commentout{
\begin{remark}
    In principle, one can encode the transition system specification itself as part of the temporal property, and avoid the need to reason about the product system, by asserting $\init \circ \globally \tau$.
    In practice, it is useful to keep them separated to  keep the number of timer variables manageable.
\end{remark}
In practice, when relying on a first-order encoding we cannot capture the intended semantics of timers (and sometimes of the verified transition system) precisely, but only a weaker, approximate, semantics $\tilde\timesemantics$, for example, the first-order semantics.
If we consider such semantics, the reduction is no longer complete, but it is still sound, so we can use it for the purpose of verification.
We formalize this in the following claim.
\begin{claim}
Let $\tilde\timesemantics$ such that $\timesemantics\subseteq \tilde\timesemantics$ and $\tilde\semantics$ such that $\tilde\semantics\subseteq\semantics$, if $\Tspec(\tilde\semantics) \times \Tspec_{\neg \prop}(\tilde\timesemantics)$ terminates then 
$\Tspec(\semantics)\models\prop$.    
\end{claim}
}
\section{Rankings}
\label{section:rankings}

In this section we discuss how we extend and use 
the framework of implicit rankings from~\cite{ImplicitRankings} to prove termination of a transition system.
\Cref{subsection:definitionAndProofRule} gives the definition of implicit rankings, generalized such that they are parameterized by an explicit soundness condition on reachable states of the system, and shows how implicit rankings are used in termination proofs. \Cref{subsection:constructors} introduces recursive constructions of implicit rankings based on the constructors of~\cite{ImplicitRankings}, extended to produce soundness conditions, and demonstrates how timers are leveraged in the construction of implicit rankings.
In \Cref{subsection:smtEncoding} we show how we discharge the proof obligations for a termination proof based on implicit rankings using validity checks in first-order logic.

\subsection{Verifying Termination with Implicit Rankings} 
\label{subsection:definitionAndProofRule}

We first introduce our extended definition of \emph{closed} implicit rankings, 
which are those used for proofs.
We then give the definition of
implicit rankings,
that are not necessarily closed,
which can be used for recursive constructions.
For the remainder of the section, we fix a first-order signature $\signature$. To simplify the notation we assume that $\signature$ is single-sorted. The generalization to many-sorted signatures, which we use in practice, is straightforward.
A closed implicit ranking is a closed formula 
$\reduced$ 
that encodes the decrease of some implicit ranking function w.r.t.\ some implicit order.
The formula $\reduced$ is defined
over a double signature 
$\signature\uplus\signature'$,
where the rank decreases between the pre-state,
captured by $\signature$,
and the post-state,
captured by $\signature'$.
In~\cite{ImplicitRankings}, the implicit order is required to be well-founded.
Instead, we relax the requirement
and extend the notion of implicit rankings
to include a soundness condition $\condition$ on states,
such that well-foundedness of the implicit order
need only be 
guaranteed for the image of the ranking function on states that satisfy $\condition$.
We use the soundness condition to specify 
requirements that are necessary to guarantee well-foundedness of the implicit order but are not directly definable in first-order logic.
The soundness conditions we use include the requirement of finiteness of some set of elements and the requirement of well-foundedness of some state relation. 

\begin{definition} [Closed Implicit Ranking]
\label{definition:closedImplicitRanking}
Given a closed formula $\reduced$ over $\signature \uplus \signature'$ 
and a class of $\signature$-structures $\condition$, we say that $\reduced$
is a \emph{closed implicit ranking} with \emph{soundness condition} $\condition$ if for any domain $\domain$, there exist a partially ordered set $(A,<)$ and a function $f \colon \structset(\signature,\domain) \to A$, such that the following hold:
\begin{enumerate}
    \item\label{impicit-ranking-item1} for every $\struct,\struct'\in \structset(\signature,\domain)$ we have  $\struct,\struct' \models \reduced \implies f(\struct) > f(\struct')$
    \item $<\textnormal{ restricted to } \{ f(\struct) \mid \struct\in \condition \} \textnormal{ is well-founded}$.
\end{enumerate}  We call 
$(A,<)$ a \emph{ranking range} for $\domain$ and $f$ a \emph{ranking function} for $\domain$.
\end{definition}
We note that for using implicit rankings in termination proofs it suffices to require \cref{impicit-ranking-item1} (or even that $f$ is defined) only for structures in $\condition$, but this does not turn out to be helpful for constructing useful implicit rankings.

We use a closed implicit ranking $\reduced$ with soundness condition $\condition$ to show that a transition system $\Tspec(\semantics)$ terminates by showing that all reachable states of $\Tspec(\semantics)$ adhere to the soundness condition $\condition$ and that $\reduced$ is satisfied by all reachable transitions.
Termination follows since if the system had an infinite trace, that would imply an infinitely decreasing sequence in a well-founded order. 
We formalize this in the following theorem:

\begin{theorem}
\label{theorem:terminationProofRule}
If every reachable state $\struct$ of  $\Tspec(\semantics)$ satisfies $\struct\in\condition$ and every reachable transition  $(\struct,\struct')$ of $\Tspec(\semantics)$ satisfies $\struct,\struct'\models \reduced$ then $\Tspec(\semantics)$ terminates.
\end{theorem}
In \Cref{subsection:smtEncoding} we discuss how we verify the requirements of \Cref{theorem:terminationProofRule} in practice, before that, we present constructions of implicit rankings. To that end, we first generalize the definition.
Similarly to~\cite{ImplicitRankings}, we observe that we can construct implicit rankings by identifying useful primitives and composing them. 
To allow such recursive constructions  which use implicit rankings as building blocks for the construction of complex implicit rankings,
\cite{ImplicitRankings} generalizes the definition of closed implicit rankings in two ways. First, implicit rankings may have free variables, called \emph{parameters},
in which case the underlying (implicit) ranking function for each domain is defined over pairs of structures and assignments to the free variables.
Parameters are useful for constructing implicit rankings that capture aggregations over the domain.
Second, 
implicit rankings in~\cite{ImplicitRankings}
additionally include formulas $\conserved$ that approximate the weak partial order $\leq$ over the ranking range. This allows recursive constructions to rely on $\leq$ as well.
We add a third generalization, which is  a formula $\minformula$ that approximates whether an element is of minimal rank.
This generalization is used to define the soundness conditions of implicit rankings constructed via aggregations over infinite domains.

\begin{definition}[Implicit Ranking]
\label{definition:implicitRanking}
An \emph{implicit ranking} with \emph{parameters} $\seq x$ and \emph{soundness condition} $\condition$ is a triple
$\rankname=\ranktuple$
where $\reduced,\conserved$ are formulas over  $\signature\uplus\signature'$ with free variables ${\seq x},\seq {x'},\minformula$ is a formula over $\signature$ with free variables $\seq x$
and $\condition$ is a class of 
structures over $\signature$,
such that for any domain $\domain$,  there exist a partially ordered set $(A,<)$ and a function $f\colon \structset(\signature,\domain)\times \assignset(\seq x,\domain) \to A$, such that the following hold:
\begin{enumerate}
\item\label{implicit-item1} $\twopair\models \reduced(\seq {x},\seq {x'}) \implies  f\pairhigh>f\pairlow$ 
\item\label{implicit-item2} $\twopair\models \conserved(\seq {x},\seq {x'}) \implies f\pairhigh\geq f\pairlow$
\item\label{implicit-item3} $\pairhigh\models \minformula \implies f\pairhigh \textnormal{ is minimal in } (A,<)$
\item $ < \textnormal{ restricted to } \{ f\pair \mid \struct\in \condition, \assign\in\assignset(\seq x,\domain) \} \textnormal{ is well-founded}$.
\end{enumerate}    
where conditions \ref{implicit-item1},\ref{implicit-item2},\ref{implicit-item3} range over any $\pairhigh,\pairlow\in \structset(\signature,\domain)\times\assignset(\seq x,\domain)$.
We call 
$(A,<)$ a \emph{ranking range} for $\domain$ and $f$ a \emph{ranking function} for $\domain$.
\end{definition}
As with closed implicit rankings, we can weaken the requirements by only requiring conditions \ref{implicit-item1},\ref{implicit-item2},\ref{implicit-item3} for $\struct,\struct'\in \condition$, but this does not turn out to be useful.

\subsection{Constructors of Implicit Rankings}
\label{subsection:constructors}

We revisit the constructors of~\cite{ImplicitRankings}.
For each, we additionally specify the definition of $\condition$ and $\minformula$ for the resulting implicit ranking. 
Moreover, we leverage the soundness condition for implicit rankings as defined in this work to generalize the finite-domain constructors of~\cite{ImplicitRankings}, which are only sound for sorts with finite domains, to allow weaker finiteness assumptions.
\ifconference
\else
We present only a subset of the constructors that are used to demonstrate the use of timers in ranking and those with non-trivial soundness conditions, we defer the rest of the constructors to \Cref{appendix:constructors}.
\fi

\begin{theorem}
\label{theorem:constructors}
All the constructors presented are sound, in the sense that whenever their arguments satisfy their assumptions, they output an implicit ranking with suitable soundness conditions.
\end{theorem}
The first constructor  we present relies on an  existing order over domain elements in the system itself to define an implicit ranking whose underlying ranking function decreases with respect to this order.
The order is specified by a relation symbol $\orderformula( y_1,y_2)$, which is required to satisfy the formula $\immutorder(\orderformula)$
that specifies that $\orderformula$ is immutable and defines a strict partial order.
The soundness condition of this implicit ranking is that the order defined by $\orderformula$ is well-founded.
We require well-foundedness as a soundness condition
as it is not possible to specify in first-order logic.
While this may seem like deferring the problem rather than solving it, we are often able to leverage the intended semantics
of the transition system,
for example the intended semantics of timers
$\timesemantics$,
to ensure that the soundness condition holds without explicitly encoding it (see \Cref{subsection:smtEncoding}).

\begin{constructor}
The \emph{position constructor} takes a term $\term$ over $\signature$ with free variables $\seq x$ 
and a binary relation symbol $\orderformula \in \signature$, returns an implicit ranking $\mathrm{Pos}(t,\orderformula)=\ranktuple$ with parameters $\seq x$ and soundness condition $\condition$ defined by:
\begin{align*}
    &{\reduced}(\seq x,\seq x') =   \immutorder(\orderformula) \wedge \orderformula(\term'(\seq x'),\term(\seq x))\\
    & {\conserved}(\seq x,\seq x') = \immutorder(\orderformula) \wedge (\orderformula( \term'(\seq x'),\term(\seq x))\vee \term'(\seq x')=\term(\seq x))\\
    & \minformula(\seq x)=\forall y. \neg \orderformula(y ,t(\seq x) ) \qquad 
    \condition = \{ \struct \in \structset(\signature) \mid \orderformula \textnormal{ is well-founded in } \struct \}
\end{align*}
\end{constructor}

\begin{example}
\label{example:examplePos}
Consider the ticket protocol of \Cref{section:motivatingExample}, after having augmented the system with timers. 
We use the timers in implicit rankings with the position constructor as follows.
The relation $\orderformula$ is the order relation 
$<$ on $\timesort$.
Taking $\term = \timer_{\mathrm{waiting}(x_0)\wedge \globally\neg\mathrm{critical}(x_0)}$, we get the closed implicit ranking capturing the decrease of  $\timer_{\mathrm{starved}}$.
Taking $\term(x) = \timer_{\globally\eventually\mathrm{scheduled}(x)}(x)$ we get an implicit rankings with parameter $x$, which needs to be aggregated to get a closed implicit ranking.  
\end{example}
The next constructor takes as input an implicit ranking $\rankname^\insup$ and a formula $\formula$ and returns an implicit ranking that refines the given implicit ranking such that states are ranked based $\rankname^\insup$ only when $\formula$ is satisfied, and states that do not satisfy $\formula$ are ranked lower than those that do. 
It follows that every state that does not satisfy $\alpha$ is minimal according to this ranking.
This constructor is useful when we can only guarantee conservation ($\leq$) of $\rankname^\insup$ when $\formula$ is satisfied.
\begin{constructor}
\label{constructor:conditional} The \emph{conditional constructor} takes an implicit ranking $\rankname^\insup=\ranktuplesup{\insup}$ with parameters $\seq x$ and soundness condition $\conditionsuper{\insup}$, and a formula $\formula(\seq x)$. 
It returns an implicit ranking $\mathrm{Cond}(\rankname^\insup,\formula)=\ranktuple$ with parameters $\seq x$ and soundness condition $\condition$ defined by:
\begin{align*}
 &{\reduced}(\seq x,\seq x') =
 (\formula(\seq x) \wedge \neg\formula'(\seq x')
 ) \vee
 (\formula(\seq x)\wedge \formula'(
\seq x'
 ) 
 \wedge
 {\reducedsuper{\insup}}(\seq x,\seq x')
 ) 
\\
&\conserved(\seq x,\seq x')=
 \neg \formula'(\seq x') \vee
 (\formula(\seq x)\wedge \formula'(
\seq x') 
 \wedge
 {\conservedsuper{\insup}}(\seq x,\seq x')
 ) \quad \minformula(\seq x)=\neg\formula(\seq x) \quad \condition = 
\conditionsuper{\insup}
\end{align*}
\end{constructor}

\begin{example}
\label{example:CondRank}
    In \Cref{example:examplePos} we show how to use the timer $\timer_{\globally\eventually\mathrm{scheduled}(x)}(x)$ in implicit rankings.
    In the ticket protocol every thread is scheduled infinitely often so the timer above is not always conserved. 
    Moreover, a thread being scheduled does not necessarily change the state, for example, when it tries to enter the critical section but does not hold $\mathrm{serv}$.
    For this purpose we want to consider the timer above only for threads that hold the currently served ticket.
    This is captured by the implicit ranking $\mathrm{Cond}(\mathrm{Pos}(\timer_{\globally\eventually\mathrm{scheduled}(x)}(x),<),\mathrm{myt}(x)=\mathrm{serv})$.
\end{example}
We continue by reintroducing the domain-based constructors of~\cite{ImplicitRankings}. 
These aggregate rankings that are parameterized by domain elements.
In~\cite{ImplicitRankings} these were defined to be sound only for finite domains. 
Here, we slightly weaken this assumption and instead introduce it as part of the soundness condition, requiring finiteness only of the set of non-minimal elements in the base ranking.
Since minimal values cannot decrease, having finitely-many non-minimal values is sufficient to maintain well-foundedness of the aggregated ranking. 

The domain-pointwise constructor lifts an implicit ranking $\rankname^\insup$ with ranking range $(P,<_P)$ to an implicit ranking whose ranking range is the set of functions $Y\to P$ where $Y$ is possibly infinite.
The set $Y\to P$ is ordered by the pointwise ordering: $a <_{\text{pw}} b$ if and only if there exists $y\in Y$ such that $a(y)<_P b(y)$ and for every $y\in Y$ we have $a(y)\leq_P b(y)$. If $<_P$ is well-founded, we can restrict the produced partial order $<_\text{pw}$ to be well-founded even if $Y$ is infinite by only considering functions $a\in Y\to P$ that have finitely many $y\in Y$ such that $a(y)$ is non-minimal. 
This is captured by the soundness condition of the constructor.

\begin{constructor}
\label{constructor:domainPointwise}
The \emph{domain-pointwise constructor} takes an implicit ranking
$\rankname^\insup = \ranktuplesup{\insup}$
with parameters $\seq x = \seq y\concatvar\seq z$ and soundness condition $\conditionsuper{\insup}$.
It returns an implicit ranking $\mathrm{DomPW}(\rankname^{\insup},\seq y) = \ranktuple$ with parameters $\seq z$ and soundness condition $\condition$ defined by:
\begin{align*}
&{\reduced}(\seq z,\seq z') = {\conserved}(\seq {z}, \seq {z'}) \wedge
(\exists \seq y. {\reducedsuper{\insup}}( \seq y \concatvar \seq {z}, \seq y \concatvar \seq {z'} ))\\
&{\conserved}(\seq z,\seq z') = \forall {\seq y}. \conservedsuper{\insup}(\seq y \concatvar \seq {z}, \seq y \concatvar \seq {z'}) 
\qquad\qquad  
\minformula(\seq z) = \forall \seq y. \minsuper{\insup}(\seq y \concatvar \seq z)\\
&{\condition} = \{  \struct \in \structset(\signature) \mid  \struct\in \conditionsuper{\insup} 
\textnormal{ and } \forall \assign\in\assignset(\seq z).\ \mathrm{finite}_{\seq y}(\neg \minsuper{\insup}(\seq y\concatvar\seq z),\struct,\assign) \}
\end{align*}
where $\mathrm{finite}_{\seq y}(\formula,\struct,\assign) := | \{
\otherassign \in \assignset(\seq y)\ \mid (\struct, \otherassign \concatfunc \assign)\models \formula(\seq y\concatvar\seq z)\} | < \infty$.
\end{constructor}
\begin{example}
\label{example:DomPWTimers}
    The implicit rankings constructed in \Cref{example:CondRank} captures the decrease of rank of a scheduling timer for some thread $x$ if it holds the currently served ticket, 
    it thus has $x$ as a parameter. 
    If we aggregate over all such threads with $\mathrm{DomPW}(\mathrm{Cond}(\mathrm{Pos}(\timer_{\globally\eventually\mathrm{scheduled}(x)}(x),<),\mathrm{myt}(x)=\mathrm{serv}),x)$
    we consider the scheduling of all such threads. 
    This pattern is the primary way we utilize timers for implicit rankings, and so we introduce ``sugar'' for it: 
    \[\mathrm{TimerRank}(\temporalformula(\seq x),\formula(\seq x))=\mathrm{DomPW}(\mathrm{Cond}(\mathrm{Pos}(\timer_{\temporalformula}(\seq x),<),\formula(\seq x),\seq x).\]
\end{example}
The final constructor lifts an implicit ranking $\rankname^\insup$ with ranking range $(P,<_P)$ to an implicit ranking whose ranking range is the set of functions $Y\to P$ where $(Y,<_Y)$ is a possibly infinite well-founded set.
The set of functions is ordered by the \emph{reverse} lexicographic ordering:
$a <_{\mathrm{lex}} b$ if and only if for every $y\in Y$ such that $a(y) \nleq_P b(y)$ there exists $y_0$ such that $y<_Y y_0$ and $a(y_0) <_P b(y_0)$.
Well-foundedness of $<_Y$ along with finiteness of the set of non-minimal elements in $\rankname^\insup$ ensure well-foundedness of the reverse lexicographic ordering.
A reverse lexicographic ordering is crucial.
To see this, consider for example the standard lexicographic ordering on $\{0,1\}^\nat$,
which is not well-founded
since it includes the sequence
$(1,0,\ldots),(0,1,0,\ldots),(0,0,1,0,\ldots),\ldots$
that is infinitely decreasing.

\begin{constructor}
\label{constructor:domainLexicographic}
The \emph{domain-lexicographic constructor}   takes an implicit ranking 
$\rankname^\insup=\ranktuplesup{\insup}$
with parameters $\seq x =  y\concatvar\seq z$ and soundness condition $\conditionsuper{\insup}$,
and a relation $\orderformula( y_1, y_2)$ over $\signature$.
It returns an implicit ranking $\mathrm{DomLex}(\rankname^{\insup},y,\orderformula) = \ranktuple$ with parameters $\seq z$ and soundness condition $\condition$ defined by:
\begin{align*}
&{\reduced}(\seq z,\seq z') = {\conserved}(\seq {z}, \seq {z'}) \wedge
(\exists y. {\reducedsuper{\insup}}(  y \concatvar \seq {z}, y \concatvar \seq {z'} )) \\ 
&{\conserved}(\seq z,\seq z') = \immutorder(\orderformula)\wedge \forall y.  (\conservedsuper{\insup}(y\concatvar \seq {z},
y\concatvar \seq {z'})
\vee \exists  y_0. (
\orderformula(y,y_0)
\wedge
\reducedsuper{\insup}( {y_0}\concatvar \seq {z},
{y_0}\concatvar \seq {z'})
) 
) \\
&\minformula(\seq z) = \forall y. \minsuper{\insup}(y \concatvar \seq z)\\
&{\condition} = 
\left\{ 
    \struct \in \structset(\signature) 
    \left|
    \begin{array}{l}
    \struct\in \conditionsuper{\insup}, \orderformula \textnormal{ is well-founded in } \struct,
    \\
    \forall \assign\in\assignset(\seq z).\ \mathrm{finite}_{y}(\neg \minsuper{\insup}( y\concatvar\seq z),\struct,\assign)
    \end{array}
    \right.
\right\}
\end{align*}
\end{constructor}

\begin{example}
\label{example:lexArray}
Consider a program that takes as input an array of natural numbers $c[0,\ldots,n-1]$, and in each iteration, an index $i$ with $c[i] >0$ is chosen and the operations $c[i]:=c[i]-1; c[i-1]=*$ are performed.
This program can be shown to be terminating with the implicit ranking $\mathrm{DomLex}(\mathrm{Pos}(c[i],<_\nat),i,<_\nat)$.
\end{example}

\subsection{Validating Decrease of Rank and Soundness Conditions}
\label{subsection:smtEncoding}

We now revisit \Cref{theorem:terminationProofRule} and show how to verify its two premises
in order to prove termination. 
Fix a signature $\Sigma$, a transition system $\Tspec(\semantics)$ with $\Tspec = \specification$ and a closed implicit ranking $\reduced$ with soundness condition $\condition$.
We need to show that any reachable state 
of $\Tspec(\semantics)$ satisfies the soundness condition $\condition$ and any reachable transition 
satisfies $\reduced$.
To discharge the premises we utilize an inductive invariant $\inv$ as a way to overapproximate the reachable states, 
for which we need to verify that $\inv$ is satisfied in the initial states and
that $\inv$ is preserved by transitions. 
Then, for the second premise, we need to verify that transitions beginning in a state satisfying $\inv$ satisfy $\reduced$. These requirements are captured by the validity of the following verification conditions in first-order logic:
\begin{inparaenum}
\item[(i)] $\init \wedge \axioms \to \inv$
\item[(ii)] $\inv \wedge \axioms \wedge \tr \wedge \axioms' \to \inv'$
\item[(iii)] $\inv \wedge \axioms \wedge \tr \wedge \axioms' \to \reduced$.
\end{inparaenum}

As for the soundness condition $\condition$, verifying that a certain structure $\struct$ satisfies $\condition$ is not always possible in first-order logic.
This is because we use the soundness condition to encode those requirements of the ranking that are not definable in first-order logic.
Specifically, the soundness conditions introduced by our constructors are of two types: well-foundedness of a relation and finitneness of a set.
These notions are not first-order definable, so we cannot reduce their verification to first-order validity checks.
Instead, we employ two complementary strategies to verify the soundness condition: relying on the intended system semantics, or finding sufficient conditions that are expressible in first-order logic.

\emph{Strategy 1.}
In some cases the intended system semantics trivially implies the validity of certain soundness conditions:
\begin{inparaenum}
    \item A set of elements of a sort with finite-domain semantics is finite.
    \item A partial order on a sort with finite-domain semantics is well-founded.
    \item A partial order that is well-founded 
    in the intended semantics of the system, 
    for example, the $<$ relation on a sort with intended semantics of the extended natural numbers (as in timers), is well-founded.
\end{inparaenum}

\emph{Strategy 2.}
The second strategy is aimed at showing finiteness of a set of elements $S$ in all reachable states, even though the domain of the relevant sort is not finite.
The approach is based on~\cite{towards_liveness_proofs}, with some generalization.
The idea is that if there is $m\in \nat$ such that initially there are at most $m$ elements in $B$ and any transition adds at most $m$ elements to $B$ then by induction, $B$ is finite in all reachable states. 
This also implies that any set $A$ such that $A \subseteq B$ in all reachable states is finite. 
In the soundness conditions we consider, the set $A$ is parametrized by some $\seq z$ and for a given $\seq z$, it is defined as the set of assignments to $\seq y$ that satisfy a formula $\formula(\seq y\concatvar \seq z)$, where our goal is to verity finiteness of $A$ for any $\seq z$.
We therefore encode the aforementioned idea by considering a formula  $\formulaa(\seq y\concatvar \seq z)$ used as an inductively-finite over-approximation of $\formula(\seq y\concatvar \seq z)$. 
This results in the following verification conditions, outlined for $m=1$, for any $\seq z$:
\begin{inparaenum}
    \item[(i)] $\formula(\seq y\concatvar \seq z) \wedge \inv \wedge  \axioms\to \formulaa(\seq y\concatvar \seq z)$;
    \item[(ii)] $\init \wedge \inv \wedge  \axioms\to \forall \seq z \ \exists \seq y_0 \ \forall \seq y \ \formulaa(\seq y\concatvar \seq z)\to \seq y=\seq y_0$; and
    \item[(iii)] $\inv \wedge \axioms\wedge \tr \wedge \axioms' \to \forall \seq z \ \exists \seq y_0 \ \forall \seq y \ \formulaa'(\seq y\concatvar \seq z)\to \seq y=\seq y_0 \vee \formulaa(\seq y\concatvar \seq z)$.
\end{inparaenum}

\begin{example}
\label{example:soundness}
    Continuing with \Cref{example:DomPWTimers}, the soundness condition can be understood as the requirement that in all reachable states there are finitely many threads $x$ for which $\formula(x) = (\mathrm{myt}(x) = \mathrm{serv})$ holds. 
    If we take $\formulaa(x)=\neg(\mathrm{idle}(x))$ and an appropriate invariant $\inv$ we can easily verify the verification conditions of strategy 2 above with $m=1$, and thus validate the soundness condition.
\end{example}
\section{Implementation and Evaluation}
\label{section:evaluation}
\begin{table}[t]
\centering
\caption{Evaluation Results.
Time is given in seconds, Con. stands for the number of constructors in the implicit ranking, 
$\mathrm{Fin}$ stands for the number of inductively-finite approximations,
Inv. stands for the number of conjuncts in the invariant, $|\mathrm{Proof}|$ stands for the number of tokens in the entire proof and $|\mathrm{L2S\ Proof}|$ stands for the number of tokens in the entire proof with the liveness-to-safety reduction in the  Ivy file.
}
\scalebox{0.9}{
\bgroup
\def\arraystretch{1.5}
\setlength\tabcolsep{4pt}
\footnotesize
\begin{tabular}{| l | c | c | c | c | c | c | c | }
\hline
Example &
Source &
Time (s) &
Con. &
$\mathrm{Fin}$ &
Inv. &
$|\mathrm{Proof}|$ &
$|\mathrm{L2S\ Proof}|$ 
\\ \hline
Ticket (\Cref{section:motivatingExample}) &
\multirow{3}{*}{\cite{liveness_to_safety}} & 
2 & 6 & 2 & 20 & 256 & 407
\\\cline{1-1}\cline{3-8}
Paxos&&
101 & 11 & 1 & 20 & 238 & 521
\\\cline{1-1}\cline{3-8}
ABP &&
4 & 31 & 2 & 34 & 539 & 860
\\\hline
HRB-Correct&
\multirow{2}{*}{\cite{BerkovitsCAV19}}&  
135 & 7 & 0 & 6 & 99 & 446
\\\cline{1-1}\cline{3-8}
HRB-Relay&&
56 & 8 & 0 & 17 & 316 & 593
\\ 
\hline
Ackermann&
\cite{power_of_temporal_prophecy}&
0.8 & 5 & 2 & 9 & 211 & 748
\\ 
\hline
LexArray (\Cref{example:lexArray})& - &
0.08 & 4 & 1 & 2 & 30 & -
\\ 
\hline
TimestampedQueue&
\multirow{3}{*}{\cite{towards_liveness_proofs}} & 
0.5 & 6 & 1 & 2 & 73 & -
\\\cline{1-1}\cline{3-8}
CascadingQueue&
&
1.2 & 9 & 2 & 7 & 175 & -
\\\cline{1-1}\cline{3-8}
ReorderingQueue&
&
2.3 & 9 & 2 & 11 & 220 & -
\\ 
\hline
MutexRing&
\multirow{7}{*}{\cite{ImplicitRankings}}& 
0.5 & 8 & 0 & 5 & 93 & -
\\\cline{1-1}\cline{3-8}
LeaderRing&&  
0.3 & 8 & 1 & 4 & 76 & -
\\\cline{1-1}\cline{3-8}
ToyStabilization&& 
0.3 & 5  & 0 & 3 & 69 & -
\\\cline{1-1}\cline{3-8}
Dijkstra $k$-State &&
64.8 & 6 & 0 & 3 & 213 & -
\\\cline{1-1}\cline{3-8}
BinaryCounter&&
0.02 & 4 & 0 & 0 & 21 & -
\\\cline{1-1}\cline{3-8}
SAT-Backtrack&&
0.3 & 9 & 0 & 8 & 157 & -
\\\cline{1-1}\cline{3-8}
SAT-CDCL&&
5.3 & 8 & 0 & 8 & 148 & -
\\ 
\hline
\end{tabular}
\egroup}
\label{table:evaluationExamples}
\end{table}

We implemented our approach for verification of temporal properties with timers and implicit rankings in a deductive verification tool.
Our tool is implemented in Python, using the Z3 API~\cite{z3}, it is available at~\cite{timers_and_rankings_artifact}. 
It takes as input a specification of a transition system in first-order logic, along with a temporal property in FO-LTL, an implicit ranking given by a composition of constructors, an inductive invariant, and, when strategy 2 of \Cref{subsection:smtEncoding} is applied, also inductively-finite approximations, all of which can use the timer variables.
The tool computes the timer reduction of \Cref{subsection:reduction},
generates the verification conditions of \Cref{subsection:smtEncoding} and discharges them using the Z3 solver. When one of the verification conditions is violated, the tool produces a counterexample in the form of a transition of the augmented system. 
In the SMT queries, the timers sort $\timesort$ is encoded using interpreted integers (with $\infty$ encoded by $-1$). The combination of uninterpreted functions, quantifiers and integers is sometimes difficult for the solver. We therefore implement an optimization where for invariants that do not involve timers we allow to check inductiveness w.r.t.\ the original system without the timers.   
Additionally, similarly to~\cite{ImplicitRankings}, we allow hint terms for existential quantifiers which help in queries with quantifier alternations.


We evaluated our tool on a set of examples of transitions systems and their properties from previous works~\cite{ImplicitRankings,liveness_to_safety,power_of_temporal_prophecy,BerkovitsCAV19,towards_liveness_proofs}, with user-provided implicit rankings and invariants. 
\ifconference
\else 
We expand on all examples in \Cref{appendix:examples}.
\fi
We ran our tool using a Macbook Pro with an Apple M1 Pro CPU, Z3 version 4.13.3.

\paragraph{Results.}
\Cref{table:evaluationExamples} summarizes our experimental results. To give a sense of the complexity of the proofs, we report for each example the time required for validating the verification conditions, the number of constructors used in the implicit ranking (where TimerRank is counted as one), the number of inductively-finite approximations needed, and the number of conjuncts in the invariant.
We additionally report the total size of the proof, measured by the number of tokens (or terms). 
Where a proof by the liveness-to-safety method exists, we also report this same metric for the corresponding proof, taken from Ivy~\cite{ivy}.

For examples from~\cite{ImplicitRankings,towards_liveness_proofs}, 
our proofs are similar to the original ones, except that in those works specific liveness proof rules were used while we prove termination of the augmented system.
The modifications of the proofs to our general framework are relatively simple, converting premises of the aforementioned rule into instantiations of the $\mathrm{TimerRank}$ constructor.
For examples that were proven by the liveness-to-safety, our proofs are simpler, as they do not require reasoning about a finite abstraction or a monitored copy of the state.
This can be observed in the significant difference in proof size in some examples.
\section{Related Work}
\label{section:relatedWork}
We have presented an approach for verifying temporal properties of systems that operate over unbounded domains, modeled in first-order logic.
Many works are aimed at verifying similar systems and properties. 
Some works take a fully automatic approach that applies to restricted systems~\cite{LTL-falsification,liveness_parametrized_programs,regular_abstractions,automataProgramTermination,liveness_invisible,newApproachTermination,ranking_functions_size_change,transition_invariants,liveness_to_safety_via_implicit_abstraction}, notably these works do not apply to the setting of systems specified in first-order logic.  
Other works provide frameworks to design distributed systems and mechanically prove their liveness~\cite{
multi_grained_specifications,StructuralTemporalLogic,LIDO,LIDO-DAG,shipwright,ironfleet,bythos,ChordLiveness,anvil} 
and finally some works provide theoretical frameworks to analyze proof mechanisms for FO-LTL~\cite{firstOrderAutomata,first_order_buchi,power_of_temporal_proofs}.
We turn to discuss the most relevant works, which, similarly to our work, provide semi-automatic methods that rely on first-order logic solvers.

Our work is based on~\cite{ImplicitRankings} and extends it in two ways.
First, while the original work targets only response properties with parameterized fairness assumptions, our work targets any FO-LTL property.
Second, the domain-based aggregations of implicit rankings in the original work assume that the relevant domain is finite. 
In our work we allow domain-based aggregations on infinite domains, as long as the set of non-minimal elements of the aggregated rankings remains finite.

In~\cite{towards_liveness_proofs} relational rankings are used to verify temporal properties.
Relational rankings are a subset of implicit rankings that are lexicographic compositions of domain-pointwise rankings. 
Similarly to our work, this work also allows systems with infinite domains, and our strategy 2 in \Cref{subsection:smtEncoding} is inspired by their enforcement of finiteness.
Unlike our unified approach,~\cite{towards_liveness_proofs} considers different temporal properties by using different proof rules for different types of properties, which is more complex and less general.

The approach of~\cite{LVR} is to semi-automatically synthesize polynomial ranking functions from integer variables in the state of the system.
These variables can be from the specification itself, cardinalities of certain sets and fairness variables introduced by the user.
These fairness variables are similar to the timers we introduce in \Cref{definition:timerSystem}, and are the inspiration for timers, but they capture only specific types of fairness assumptions and are introduced manually, so they do not have formal guarantees.
Like us, they allow the user to denote certain sorts as finitely interpreted, but their treatment of finiteness is less complete.
While some examples are shared between the papers, they make some manual manipulations to fairness assumptions in non-trivial ways that significantly simplify the proofs.

In~\cite{liveness_to_safety,prophecy,power_of_temporal_prophecy} a liveness-to-safety reduction is employed to reduce the verification of temporal properties to safety verification.
The reduction performs a dynamic finite abstraction of the system.
One then proves acyclicity of the abstracted system by providing an invariant for a monitored system where the current state is compared to a saved state modulo the abstraction.
The liveness-to-safety reduction is sound but not complete, and requires employing temporal prophecy to extend its power.
Moreover, writing invariants for the augmented system may be more complicated 
and less intuitive
than giving implicit rankings.

\commentout{
Doesn't fit here, i added the citation in the beginning.\\
The reified time construction takes a formula in a temporal logic and translates into a formula in a ``pure'' logic that is equisatisfiable. For FO-LTL this was done in~\cite{power_of_temporal_proofs}, translating an FO-LTL formula $\varphi$ to a formula $\varphi_\timename$ in first-order logic with the theory of natural numbers such that $\varphi$ is satisfiable if and only if $\varphi$ is satisfiable under the intended semantics of natural numbers.
This reduction is sound and complete and was used to explore the power of proof systems for FO-LTL, defining weak completeness criterion, based on the resulting proof system for first-order logic with natural number semantics.
This approach has not been implemented and we beleive it is less useful for deductive proofs and automation.
\raz{im still not sure, maybe should be earlier.}
}

\paragraph{\bf Acknowledgement}
We thank the anonymous reviewers and Eden Frenkel for helpful comments.
The research leading to these results has received funding from the
European Research Council under the European Union's Horizon 2020 research and innovation programme (grant agreement No [759102-SVIS]).
This research was partially supported by the Israeli Science Foundation (ISF) grant No.\ 2117/23, by a research grant from the Center for New Scientists at the Weizmann Institute of Science and by a grant from the Azrieli Foundation.

\paragraph{\bf Data-Availability Statement}
The paper is accompanied by an artifact containing the software and data required to reproduce the results. 
A snapshot of the artifact corresponding to the version used in this paper is available at~\cite{timers_and_rankings_artifact}.
For reuse and future development, the source code is available at~\cite{timers_and_rankings_repo}.

\bibliography{misc/references}

@misc{timers_and_rankings_artifact,
  author       = {Elad, Neta and
                  Lotan, Raz},
  title        = {Verifying First-Order Temporal Properties of
                   Infinite States Systems via Timers and
                   Rankings (Artifact)
                  },
  month        = jan,
  year         = 2026,
  publisher    = {Zenodo},
  url          = {https://doi.org/10.5281/zenodo.18157346},
}

@misc{timers_and_rankings_repo,
  title = {implicit-rankings-timers},
  url   = {https://github.com/neta-elad/implicit-rankings-timers},
  author = {Neta Elad and Raz Lotan}
}

@inproceedings{ivy,
  author    = {Oded Padon and Kenneth L. McMillan and Aurojit Panda and Mooly Sagiv and Sharon Shoham},
  title     = {Ivy: Safety Verification by Interactive Generalization},
  booktitle = {PLDI '16: Proceedings of the 37th ACM SIGPLAN Conference on Programming Language Design and Implementation},
  year      = {2016},
  month     = jun,
  pages     = {614--630},
  doi       = {10.1145/2908080.2908118}
}

@article{MannaDershowitz,
author = {Dershowitz, Nachum and Manna, Zohar},
title = {Proving termination with multiset orderings},
year = {1979},
issue_date = {Aug. 1979},
publisher = {Association for Computing Machinery},
address = {New York, NY, USA},
volume = {22},
number = {8},
issn = {0001-0782},
url = {https://doi.org/10.1145/359138.359142},
journal = {Commun. ACM},
month = aug,
pages = {465–476},
numpages = {12}
}

@inproceedings{z3,
  author       = {Leonardo Mendon{\c{c}}a de Moura and
                  Nikolaj S. Bj{\o}rner},
  title        = {{Z3:} An Efficient {SMT} Solver},
  booktitle    = {Tools and Algorithms for the Construction and Analysis of Systems,
                  14th International Conference, {TACAS} 2008, 
                  },
  year         = {2008},
  doi          = {10.1007/978-3-540-78800-3\_24},
}

@inproceedings{towards_liveness_proofs,
  author       = {Kenneth L. McMillan},
  title        = {Toward Liveness Proofs at Scale},
  booktitle    = {Computer Aided Verification - 36th International Conference, {CAV}
                  2024, Montreal, QC, Canada, July 24-27, 2024, Proceedings, Part {I}},
  series       = {Lecture Notes in Computer Science},
  volume       = {14681},
  pages        = {255--276},
  publisher    = {Springer},
  year         = {2024},
  doi          = {10.1007/978-3-031-65627-9\_13},
  bibsource    = {dblp computer science bibliography, https://dblp.org}
}

@inproceedings{transition_invariants,
  author       = {Andreas Podelski and
                  Andrey Rybalchenko},
  title        = {Transition Invariants},
  booktitle    = {19th {IEEE} Symposium on Logic in Computer Science {(LICS} 2004),
                  14-17 July 2004, Turku, Finland, Proceedings},
  pages        = {32--41},
  publisher    = {{IEEE} Computer Society},
  year         = {2004},
  doi          = {10.1109/LICS.2004.1319598},
  bibsource    = {dblp computer science bibliography, https://dblp.org}
}

@article{liveness_invisible,
  author       = {Yi Fang and
                  Nir Piterman and
                  Amir Pnueli and
                  Lenore D. Zuck},
  title        = {Liveness with invisible ranking},
  journal      = {Int. J. Softw. Tools Technol. Transf.},
  volume       = {8},
  number       = {3},
  pages        = {261--279},
  year         = {2006},
  doi          = {10.1007/S10009-005-0193-X},
  bibsource    = {dblp computer science bibliography, https://dblp.org}
}

@article{liveness_to_safety,
  author       = {Oded Padon and
                  Jochen Hoenicke and
                  Giuliano Losa and
                  Andreas Podelski and
                  Mooly Sagiv and
                  Sharon Shoham},
  title        = {Reducing liveness to safety in first-order logic},
  journal      = {Proc. {ACM} Program. Lang.},
  volume       = {2},
  number       = {{POPL}},
  pages        = {26:1--26:33},
  year         = {2018},
  doi          = {10.1145/3158114},
}

@article{LVR,
  author       = {Jianan Yao and
                  Runzhou Tao and
                  Ronghui Gu and
                  Jason Nieh},
  title        = {Mostly Automated Verification of Liveness Properties for Distributed
                  Protocols with Ranking Functions},
  journal      = {Proc. {ACM} Program. Lang.},
  volume       = {8},
  number       = {{POPL}},
  pages        = {1028--1059},
  year         = {2024},
  doi          = {10.1145/3632877},
}

@article{prophecy,
  author       = {Oded Padon and
                  Jochen Hoenicke and
                  Kenneth L. McMillan and
                  Andreas Podelski and
                  Mooly Sagiv and
                  Sharon Shoham},
  title        = {Temporal prophecy for proving temporal properties of infinite-state
                  systems},
  journal      = {Formal Methods Syst. Des.},
  volume       = {57},
  number       = {2},
  pages        = {246--269},
  year         = {2021},
  doi          = {10.1007/S10703-021-00377-1},
}

@article{dijkstra_self_stab,
author = {Dijkstra, Edsger W.},
title = {Self-stabilizing systems in spite of distributed control},
year = {1974},
issue_date = {Nov. 1974},
publisher = {Association for Computing Machinery},
address = {New York, NY, USA},
volume = {17},
number = {11},
issn = {0001-0782},
url = {https://doi.org/10.1145/361179.361202},
journal = {Commun. ACM},
month = {nov},
pages = {643–644},
numpages = {2}
}

@inproceedings{BerkovitsCAV19,
  author       = {Idan Berkovits and
                  Marijana Lazic and
                  Giuliano Losa and
                  Oded Padon and
                  Sharon Shoham},
  title        = {Verification of Threshold-Based Distributed Algorithms by Decomposition
                  to Decidable Logics},
  booktitle    = {Computer Aided Verification - 31st International Conference, {CAV}
                  2019},
  url          = {https://doi.org/10.1007/978-3-030-25543-5\_15},
  year         = {2019},
}

@article{regular_abstractions,
  author       = {Chih{-}Duo Hong and
                  Anthony W. Lin},
  title        = {Regular Abstractions for Array Systems},
  journal      = {Proc. {ACM} Program. Lang.},
  volume       = {8},
  number       = {{POPL}},
  pages        = {638--666},
  year         = {2024},
  url          = {https://doi.org/10.1145/3632864},
  bibsource    = {dblp computer science bibliography, https://dblp.org}
}

@article{chang_roberts,
  author       = {Ernest J. H. Chang and
                  Rosemary Roberts},
  title        = {An Improved Algorithm for Decentralized Extrema-Finding in Circular
                  Configurations of Processes},
  journal      = {Commun. {ACM}},
  volume       = {22},
  number       = {5},
  pages        = {281--283},
  year         = {1979},
  url          = {https://doi.org/10.1145/359104.359108},
  bibsource    = {dblp computer science bibliography, https://dblp.org}
}

@inproceedings{liveness_to_safety_via_implicit_abstraction,
  author       = {Jakub Daniel and
                  Alessandro Cimatti and
                  Alberto Griggio and
                  Stefano Tonetta and
                  Sergio Mover},
  title        = {Infinite-State Liveness-to-Safety via Implicit Abstraction and Well-Founded
                  Relations},
  booktitle    = {Computer Aided Verification - 28th International Conference, {CAV}
                  2016, Toronto, ON, Canada, July 17-23, 2016},
  year         = {2016},
  url          = {https://doi.org/10.1007/978-3-319-41528-4\_15},
}

@article{hrb_source,
  author       = {Josef Widder and
                  Ulrich Schmid},
  title        = {Booting clock synchronization in partially synchronous systems with
                  hybrid process and link failures},
  journal      = {Distributed Comput.},
  volume       = {20},
  number       = {2},
  pages        = {115--140},
  year         = {2007},
  url          = {https://doi.org/10.1007/s00446-007-0026-0},
  doi          = {10.1007/S00446-007-0026-0},
  bibsource    = {dblp computer science bibliography, https://dblp.org}
}

@article{ironfleet,
  author       = {Chris Hawblitzel and
                  Jon Howell and
                  Manos Kapritsos and
                  Jacob R. Lorch and
                  Bryan Parno and
                  Michael Lowell Roberts and
                  Srinath T. V. Setty and
                  Brian Zill},
  title        = {IronFleet: proving safety and liveness of practical distributed systems},
  journal      = {Commun. {ACM}},
  volume       = {60},
  number       = {7},
  pages        = {83--92},
  year         = {2017},
  url          = {https://doi.org/10.1145/3068608},
  bibsource    = {dblp computer science bibliography, https://dblp.org}
}

@inproceedings{anvil,
  author       = {Xudong Sun and
                  Wenjie Ma and
                  Jiawei Tyler Gu and
                  Zicheng Ma and
                  Tej Chajed and
                  Jon Howell and
                  Andrea Lattuada and
                  Oded Padon and
                  Lalith Suresh and
                  Adriana Szekeres and
                  Tianyin Xu},
  title        = {Anvil: Verifying Liveness of Cluster Management Controllers},
  booktitle    = {18th {USENIX} Symposium on Operating Systems Design and Implementation},
    year         = {2024},
  url          = {https://www.usenix.org/conference/osdi24/presentation/sun-xudong},
}

@article{ranking_functions_size_change,
  author       = {Amir M. Ben{-}Amram and
                  Chin Soon Lee},
  title        = {Ranking Functions for Size-Change Termination {II}},
  journal      = {Log. Methods Comput. Sci.},
  volume       = {5},
  number       = {2},
  year         = {2009},
  url          = {http://arxiv.org/abs/0903.4382},
  bibsource    = {dblp computer science bibliography, https://dblp.org}
}

@inproceedings{automataProgramTermination,
  author       = {Yu{-}Fang Chen and
                  Matthias Heizmann and
                  Ondrej Leng{\'{a}}l and
                  Yong Li and
                  Ming{-}Hsien Tsai and
                  Andrea Turrini and
                  Lijun Zhang},
  title        = {Advanced automata-based algorithms for program termination checking},
  booktitle    = {Programming
                  Language Design and Implementation, {PLDI} 2018},
  url          = {https://doi.org/10.1145/3192366.3192405},
  bibsource    = {dblp computer science bibliography, https://dblp.org},
    year=2018
}

@article{LIDO-DAG,
author = {Qiu, Longfei and Xiao, Jingqi and Shin, Ji-Yong and Shao, Zhong},
title = {LiDO-DAG: A Framework for Verifying Safety and Liveness of DAG-Based Consensus Protocols},
year = {2025},
issue_date = {June 2025},
publisher = {Association for Computing Machinery},
address = {New York, NY, USA},
volume = {9},
number = {PLDI},
url = {https://doi.org/10.1145/3729306},
journal = {Proc. ACM Program. Lang.},
month = jun,
articleno = {203},
numpages = {25}
}

@inproceedings{liveness_parametrized_programs,
  author       = {Azadeh Farzan and
                  Zachary Kincaid and
                  Andreas Podelski},
  title        = {Proving Liveness of Parameterized Programs},
  booktitle    = {Annual {ACM/IEEE} Symposium on Logic in Computer
                  Science, {LICS} '16},
  year         = {2016},
  doi          = {10.1145/2933575.2935310},
  bibsource    = {dblp computer science bibliography, https://dblp.org}
}

@inproceedings{automata_approach,
  author       = {Moshe Y. Vardi},
  title        = {An Automata-Theoretic Approach to Linear Temporal Logic},
  booktitle    = {Logics for Concurrency - Structure versus Automata},
  doi          = {10.1007/3-540-60915-6\_6},
year=1995
}

@inproceedings{CSP_SAT,
  title={Using CSP look-back techniques to solve real-world SAT instances},
  author={Bayardo Jr, Roberto J and Schrag, Robert},
  booktitle={Aaai/iaai},
  pages={203--208},
  year={1997},
  organization={Citeseer}
}

@article{power_of_temporal_proofs,
  author       = {Mart{\'{\i}}n Abadi},
  title        = {The Power of Temporal Proofs},
  journal      = {Theor. Comput. Sci.},
  volume       = {65},
  number       = {1},
  pages        = {35--83},
  year         = {1989},
  url          = {https://doi.org/10.1016/0304-3975(89)90138-2},
  bibsource    = {dblp computer science bibliography, https://dblp.org}
}

@inproceedings{multi_grained_specifications,
author = {Ouyang, Lingzhi and Sun, Xudong and Tang, Ruize and Huang, Yu and Jivrajani, Madhav and Ma, Xiaoxing and Xu, Tianyin},
title = {Multi-Grained Specifications for Distributed System Model Checking and Verification},
year = {2025},
isbn = {9798400711961},
doi = {10.1145/3689031.3696069},
booktitle = {Proceedings of the Twentieth European Conference on Computer Systems},
series = {EuroSys '25}
}

@misc{StructuralTemporalLogic,
      title={Structural temporal logic for mechanized program verification}, 
      author={Eleftherios Ioannidis and Yannick Zakowski and Steve Zdancewic and Sebastian Angel},
      year={2025},
      eprint={2410.14906},
      archivePrefix={arXiv},
      primaryClass={cs.PL},
      url={https://arxiv.org/abs/2410.14906}, 
}

@article{first_order_buchi,
  author       = {Wenhui Zhang},
  title        = {First order B{\"{u}}chi automata and their application to verification
                  of {LTL} specifications},
  journal      = {J. Log. Algebraic Methods Program.},
  volume       = {142},
  pages        = {101021},
  year         = {2025},
  doi          = {10.1016/J.JLAMP.2024.101021},
  bibsource    = {dblp computer science bibliography, https://dblp.org}
}

@inproceedings{ImplicitRankings,
  author       = {Raz Lotan and
                  Sharon Shoham},
  title        = {Implicit Rankings for Verifying Liveness Properties in First-Order
                  Logic},
  booktitle    = {Tools and Algorithms for the Construction and Analysis of Systems
                  {TACAS} 2025},
  year         = {2025},
  doi          = {10.1007/978-3-031-90643-5\_20},
}

@article{LTL-falsification,
title = {{LTL} falsification in infinite-state systems},
journal = {Information and Computation},
year = {2022},
url = {https://www.sciencedirect.com/science/article/pii/S0890540122001328},
author = {Alessandro Cimatti and Alberto Griggio and Enrico Magnago},
}

@article{LIDO,
  author       = {Longfei Qiu and
                  Yoonseung Kim and
                  Ji{-}Yong Shin and
                  Jieung Kim and
                  Wolf Honor{\'{e}} and
                  Zhong Shao},
  title        = {LiDO: Linearizable Byzantine Distributed Objects with Refinement-Based
                  Liveness Proofs},
  journal      = {Proc. {ACM} Program. Lang.},
  volume       = {8},
  number       = {{PLDI}},
  pages        = {1140--1164},
  year         = {2024},
  url          = {https://doi.org/10.1145/3656423},
  bibsource    = {dblp computer science bibliography, https://dblp.org}
}

@misc{shipwright,
      title={Shipwright: Proving liveness of distributed systems with Byzantine participants}, 
      author={Derek Leung and Nickolai Zeldovich and Frans Kaashoek},
      year={2025},
      eprint={2507.14080},
      archivePrefix={arXiv},
      primaryClass={cs.DC},
      url={https://arxiv.org/abs/2507.14080}, 
}

@inproceedings{bythos,
author = {Zhao, Qiyuan and P\^{\i}rlea, George and Grzeszkiewicz, Karolina and Gilbert, Seth and Sergey, Ilya},
title = {Compositional Verification of Composite Byzantine Protocols},
year = {2024},
isbn = {9798400706363},
publisher = {Association for Computing Machinery},
address = {New York, NY, USA},
doi = {10.1145/3658644.3690355},
booktitle = {Conference on Computer and Communications Security},
location = {Salt Lake City, UT, USA},
series = {CCS '24}
}

@inproceedings{ChordLiveness,
  author       = {Julien Brunel and
                  David Chemouil and
                  Jeanne Tawa},
  title        = {Analyzing the Fundamental Liveness Property of the Chord Protocol},
  booktitle    = {Formal Methods in Computer Aided Design, {FMCAD} 2018},
  year         = {2018},
  url          = {https://doi.org/10.23919/FMCAD.2018.8603001},
  bibsource    = {dblp computer science bibliography, https://dblp.org}
}

@inproceedings{firstOrderAutomata,
author = {Geatti, Luca and Gianola, Alessandro and Gigante, Nicola},
title = {First-order automata},
year = {2025},
isbn = {978-1-57735-897-8},
url = {https://doi.org/10.1609/aaai.v39i14.33638},
booktitle = {Proceedings of the Thirty-Ninth AAAI Conference on Artificial Intelligence},
articleno = {1665},
numpages = {9},
}

@article{lamportBakery,
author = {Lamport, Leslie},
title = {A new solution of Dijkstra's concurrent programming problem},
year = {1974},
publisher = {Association for Computing Machinery},
address = {New York, NY, USA},
volume = {17},
number = {8},
issn = {0001-0782},
url = {https://doi.org/10.1145/361082.361093},
journal = {Commun. ACM},
pages = {453–455},
numpages = {3}
}

@InProceedings{newApproachTermination,
author="Herrmann, Roland
and R{\"u}mmer, Philipp",
title="A New Approach for Showing Termination of Parameterized Transition Systems",
booktitle="Implementation and Application of Automata",
year="2026",
url="https://doi.org/10.1007/978-3-032-02602-6_14",
isbn="978-3-032-02602-6"
}

@Inbook{power_of_temporal_prophecy,
author="Hoenicke, Jochen
and Padon, Oded
and Shoham, Sharon",
title="On the Power of Temporal Prophecy",
bookTitle="On the Pursuit of Insight and Elegance: Essays Dedicated to Andreas Podelski on the Occasion of His 65th Birthday",
year="2026",
publisher="Springer Nature Switzerland",
address="Cham",
pages="40--53",
isbn="978-3-032-13711-1",
url="https://doi.org/10.1007/978-3-032-13711-1_3"
}

@article{PaxosLamport,
author = {Lamport, Leslie},
title = {The part-time parliament},
year = {1998},
issue_date = {May 1998},
publisher = {Association for Computing Machinery},
address = {New York, NY, USA},
volume = {16},
number = {2},
issn = {0734-2071},
url = {https://doi.org/10.1145/279227.279229},
journal = {ACM Trans. Comput. Syst.},
month = may,
pages = {133–169},
numpages = {37}
}

@article{ackermann1928,
  title={Zum hilbertschen aufbau der reellen zahlen},
  author={Ackermann, Wilhelm},
  journal={Mathematische Annalen},
  volume={99},
  number={1},
  pages={118--133},
  year={1928},
  publisher={Springer-Verlag Berlin/Heidelberg}
}

@article{iterative_ackermann,
  author       = {Jerrold W. Grossman and
                  R. Suzanne Zeitman},
  title        = {An Inherently Iterative Computation of Ackermanns's Function},
  journal      = {Theor. Comput. Sci.},
  volume       = {57},
  pages        = {327--330},
  year         = {1988},
  url          = {https://doi.org/10.1016/0304-3975(88)90046-1},
  bibsource    = {dblp computer science bibliography, https://dblp.org}
}

@article{alternating_bit_protocol,
author = {Bartlett, K. A. and Scantlebury, R. A. and Wilkinson, P. T.},
title = {A note on reliable full-duplex transmission over half-duplex links},
year = {1969},
issue_date = {May 1969},
publisher = {Association for Computing Machinery},
address = {New York, NY, USA},
volume = {12},
number = {5},
issn = {0001-0782},
url = {https://doi.org/10.1145/362946.362970},
journal = {Commun. ACM},
month = may,
pages = {260–261},
numpages = {2},
}

\ifconference
\else
\newpage
\appendix
\section{Proofs}
\label{appendix:proofs}

\subsubsection{Proof of \Cref{lemma:timerTemporalEquivalence}}

\begin{proof}
Proof by induction on the structure of FO-LTL formulas.
For an atomic formula $p(\seq x)$, this follows from the axiom: $\axioms_p = \timer_p (\seq x) = 0 \leftrightarrow p$. 
For $\temporalformula = \eventually \temporalformulaa$: if $\hat{\struct}_i,\assign \models \timer_{\eventually \temporalformulaa} (\seq x) = 0$ then from $\axioms_{\eventually\temporalformulaa}$ the interpretation of $\timer_\temporalformulaa$ is $k<\infty$,
from the decrease of timers in transitions it follows that $\hat{\struct}_{i+1}$ interprets $\timer_\temporalformula$ to $k-1$; repeating this argument $k$ times gives us that
$\hat{\struct}_{i+k},\assign\models \timer_\temporalformulaa (\seq x) = 0$. 
Then from the induction hypothesis, $\hat{\pi}^{i+k},\assign\models\temporalformulaa$ and so $\hat{\pi}^i,\assign\models   {\eventually\temporalformulaa}$ from semantics of $\eventually$.
In the other direction, if $\hat{\pi}^i,\assign \models \eventually \temporalformulaa$ then there exists $k\in\nat$ such that $\hat{\pi}^{i+k},\assign\models \temporalformulaa$ so from induction hypothesis we have $\hat{\struct}_{i+k},\assign \models (\timer_{\temporalformulaa}=0)$, from $\axioms_{\eventually\temporalformulaa}$ we have $\hat{\struct}_{i+k},\assign\models (\timer_{\eventually \temporalformulaa}=0)$.
Now, we iteratively apply $\tr_{\eventually\temporalformulaa}$, and get that for all $j\leq i+k$ we have $\hat{\struct}_{j},\assign\models (\timer_{\eventually \temporalformulaa}=0)$, and for $j=i$ we get the desired result.
For $\temporalformula = \globally \temporalformulaa$: if $\hat{\struct}_i,\assign \models \timer_{\globally \temporalformulaa} (\seq x) = 0$ then from $\tr_{\globally\temporalformulaa}$ we have $\hat{\struct}_i,\assign \models \timer_{\temporalformulaa} (\seq x) = 0$ and $\hat{\struct}_{i+1},\assign \models t'_{\globally\temporalformulaa} (\seq x) = 0$, we can repeat this argument for $i+1$ etc., 
and get that for every $k\geq i$ we have $\hat{\struct}_i,\assign \models \timer_{\temporalformulaa} (\seq x) = 0$, 
then from induction hypothesis we get that $\hat{\pi}^k,\assign\models\temporalformulaa$ and $\hat{\pi}^i,\assign\models \globally\temporalformulaa$ from the semantics of $\globally$.
In the other direction, if $\hat{\pi}^i,\assign \models \globally \temporalformulaa$ then for every $k\in\nat$ we have $\hat{\pi}^{i+k},\assign\models \temporalformulaa$ so from induction hypothesis we have $\hat{\struct}_{i+k},\assign \models (\timer_{\temporalformulaa}=0)$.
Observe the value of  $\timer_{\neg \temporalformulaa}$ in $\hat{\struct}_i,\assign$,
if it is $k<\infty$ then from reduction of timers, for $\hat{\struct}_{i+k}$ we will have $(\timer_{\neg \temporalformulaa}=0)$, 
and by $\axioms_{\neg\temporalformulaa}$ we get $\timer_{\temporalformulaa}\neq 0$ contradicting the above.
Then we must have  $\hat{\struct}_i,\assign\models (\timer_{\neg\temporalformulaa}=\infty)$, then from $\axioms_{\globally\temporalformulaa}$, $\hat{\struct}_i,\assign\models (\timer_{\globally\temporalformulaa} =0)$ as required.
The inductive cases for the other operators are simple, and follow from first-order reasoning.
\qed
\end{proof}

\subsubsection{Proof of \Cref{lemma:traceExtension}}

\begin{proof}
To construct the trace $\hat\pi$ we need to extend each structure $\struct_i$ to a structure $\hat\struct$ over $\signature_\prop$ in a way that the timer specification and intended semantics are respected. 
First, for each structure $\struct_i$ extend it by interpreting $\timesort,\leq,0,\infty,-1$ to the extended natural numbers $\natinf$ with the usual order.
It remains to define the interpretation of the timer functions for each formula in $ \sub{\prop}$.
Let $\temporalformula\in \sub{\prop}$ with free variables $\seq x$ and a
structure $\hat{\struct}_i$, define the interpretation of $\timer_\temporalformula$ to be a function from the sorts of  $\seq x$ to $\timesort$ such that for any sequence of sorted elements 
$\seq d$ and assignment $\assign = \seq x\mapsto \seq d$, we have that $\timer_\temporalformula$ maps $\seq d$ to the the minimal $j \in \nat$ such that $ \pi^{i+j},\assign \models \temporalformula$ or $\infty$ if no such $j$ exists.
Because $\pi$ and $\hat{\pi}$ agree on symbols from $\signature$ we can use their satisfactions interchangeably.

It remains to show that the timer specification is satisfied by these interpretations. 
First, $\pi\models\prop$, it follows that for $\hat{\struct}_0$ we have $\timer_{\prop}=0$, satisfying $\init_\prop$.
Second, we show the axioms $\axioms_\prop$ on all states $\hat{\struct}_i$ which means $\axioms_\temporalformula$ is satisfied for each $\temporalformula\in\sub{\prop}$:
for atomic $p$, we have $\hat{\struct}_i\models \timer_p (\seq x) = 0$ if and only if then $\pi^i\models p$ if and only if $\hat{\struct}_i\models p$, it follows that $t_p (\seq x) = 0\leftrightarrow p$ is valid in all states $\hat{\struct}_i$.
For $\temporalformula = \eventually \temporalformulaa$: $\hat{\struct}_i\models (\timer_{\eventually \temporalformulaa}=0)$ if and only if $ \pi^i \models {\eventually \temporalformulaa} $ if and only if there exists $k\in \nat$ such that $\pi^{i+k}\models \temporalformulaa$ if and only if $\hat{\struct}_i \models \timer_\temporalformulaa < \infty$.
For $\temporalformula = \globally\temporalformulaa$, we have $\hat{\struct}_i \models (\timer_{\globally\temporalformulaa}=0)$ if and only if $\pi^i \models\globally \temporalformulaa$ if and only if for all $k \in \nat$ we have $ \pi^{i+k}\models \temporalformulaa $, if and only if for every $k\in \nat$ $\pi^{i+k}\nvDash \neg \temporalformulaa$ if and only $\hat{\struct}_i\models \timer_{\neg\temporalformulaa}<\infty$.
The satisfaction of $\axioms_\temporalformula$ for formulas constructed by the other operators are simple, and follow from first-order reasoning.
Finally, we wish to show that successive states of $\hat{\pi}$ satisfy $\tr_\prop$, for any subformula $\temporalformula$, if $\hat{\struct}_i\models  0 < \timer_\temporalformula <\infty$, let $j$ be the interpretation of $\timer_\temporalformula$ in $\hat{\struct}_i$, it is minimal such that $\pi^{i+j}\models \temporalformula$, and so $\timer_\temporalformula$ is interpreted in $\hat{\struct}$ to $j-1$, and so $(\hat\struct_i,\hat\struct_{i+1})\models \timer_\temporalformula'(\seq x) = \timer_\temporalformula(\seq x) -1$.
Moreover, if $\hat\struct_i$ interprets $\timer_\temporalformula(\seq x)$ to $\infty$ then $\pi^{i+j},\assign\nvDash \temporalformula$ for any $j$,
it follows that $\pi^{i+1+j},\assign\nvDash \temporalformula(\seq x)$
for any $j$, and so $\hat\struct_{i+1}$ interprets $\timer'_\temporalformula(\seq x)$ to $\infty$ as well, and so 
$(\hat\struct_i,\hat\struct_{i+1})\models \timer_\temporalformula'(\seq x) = \infty \to  \timer_\temporalformula(\seq x)= \infty$ as necessary.
For $\temporalformula = \eventually \temporalformulaa$, $\hat{\struct}_i\models (\timer_{\eventually\temporalformulaa}=0)$ if and only if there exists a minimal $k\in \nat$ such that $\pi^{i+k}\models \temporalformulaa$,  $k=0$ if and only if  
$\timer_\temporalformulaa (\seq x) = 0$, otherwise $k>0$ if and only if $\pi^{i+1}\models \eventually \temporalformulaa$ if and only if $\hat\struct_{i+1}\models (t'_{\eventually\temporalformulaa}=0)$, thus this scenario is equivalent to $(\hat\struct_i,\hat\struct_{i+1})\models \timer_\temporalformulaa (\seq x) = 0 \vee t'_{\eventually \temporalformulaa}$.
For $\temporalformula = \globally \temporalformulaa$, $\hat{\struct}_i\models (\timer_{\globally\temporalformulaa}=0)$ if and only if for every $k\in \nat$ we have $\pi^{i+k}\models \temporalformula$, if and only for $\pi^i\models\temporalformulaa$ and $\pi^{i+1}\models \globally \temporalformulaa$ if and only if $(\hat\struct_i,\hat\struct_{i+1})\models \timer_\temporalformulaa (\seq x) = 0 \wedge t'_{\globally \temporalformulaa}=0$.  
\qed
\end{proof}

\subsubsection{Proof of \Cref{theorem:reductionSoundComplete}}

\begin{proof}
Soundness: Assume $\Tspec(\semantics) \times \Tspec_{\neg \prop}(\timesemantics)$ terminates.
Let $\pi$ be a trace of $\Tspec(\semantics)$, we wish to show that $\pi \models \prop$, assume towards contradiction that $\pi \models \neg \prop$, 
let $\hat{\pi}$ be the extended trace of $\Tspec_{\neg \prop}(\timesemantics)$ from \Cref{lemma:traceExtension} such that $\hat{\pi}|_{\signature} = \pi$, it follows that $\hat{\pi}$ is a trace of $\Tspec(\semantics) \times \Tspec_{\neg \prop}(\timesemantics)$, contradicting the assumption that this system terminates.
Completeness: Assume $\Tspec(\semantics)$ satisfies $\prop$. 
Assume towards contradiction that $\hat{\pi}=\hat{\struct}_0,\hat{\struct}_1,\ldots$ is a trace of $\Tspec(\semantics) \times \Tspec_{\neg \prop}(\timesemantics)$. 
By the initial state axiom we have $\hat{\struct}_0 \models (\timer_{\neg \prop}=0)$ and so by \Cref{lemma:timerTemporalEquivalence} we have $\hat{\pi} \models \neg \prop$. 
Consider the projection $\pi=\hat{\pi}|_{\signature}$, this is necessarily a trace of $\Tspec(\semantics)$ and $\pi\models \neg\prop$, contradicting the assumption.
\qed
\end{proof}

\subsubsection{Proof of \Cref{theorem:terminationProofRule}}

\begin{proof}
    Assume towards contradiction that $\Tspec(\semantics)$ does not terminate, then there exists a trace $\pi=(s_i)_{i=1}^\infty$ over some shared domain $\domain$. 
    It follows that 
    $\reduced$ is a closed implicit ranking with soundness condition $\condition$, so for $\domain$ there exist a ranking range $(A,<)$ and a ranking function $f\colon \structset(\signature,\domain)\to A$.
    Observe the sequence $\left(f(s_i)\right)_{i=0}^\infty$, for every $i\in \nat$, $s_i$ is reachable, and $(s_i,s_{i+1})$ is a reachable transition, and so from assumption for all $i\in\nat$ we have $s_i \in\condition$ and $(s_{i},s_{i+1})\models\reduced$, from definition we have $f(s_{i+1})<f(s_i)$, because  
    $f(s_i)\in\{ f(\struct) \mid \struct\in \condition \}$, it follows that the sequence $\left(f(s_i)\right)_{i=0}^\infty$ is an infinitely decreasing sequence in a well-founded partial order, leading to a contradiction.
    \qed
\end{proof}

\subsubsection{Proof of \Cref{theorem:constructors} }

We prove soundness of \Cref{constructor:domainPointwise} and \Cref{constructor:domainLexicographic}, the soundness of the other constructors is similar to the soundness proofs in~\cite{ImplicitRankings}. 
Before proving soundness of the constructors we prove a claim that shows the underlying order is well-founded.
\begin{claim}
Given a well-founded set $(P,<_P)$ and a set $Y$, $(\mathrm{Fin}(Y\to P),<_{\mathrm{pw}})$, the set of functions from $Y$ to $P$  with finitely many non-minimal images, ordered by the pointwise ordering is well-founded.
\end{claim}

\begin{proof}
Assume towards contradiction that there is an infinite sequence $(f_i)_{i=0}^\infty$ such that for all $i$, $f_i\in \mathrm{Fin}(Y\to P)$ and $f_i >_{\mathrm{pw}} f_{i+1}$.
Let $S = \{ y\in Y \mid f_0(y) \text{ is non-minimal in }P\}$, define a sequence $(y_i)_{i=0}^\infty$ such that $f_{i+1}(y_i) <_P f_i(y_i)$.
It follows that for all $i\in \nat$ we have $y_i\in S $ as $f_{i+1}(y_i) <_P f_{i}(y_i) \leq_P f_0(y_i)$.
$S$ is finite, so there exists $y_\star$ such that it appears infinitely often in the sequence $(y_i)_{i=0}^\infty$ and so $f_{i+1}(y_\star) <_P f_i(y_\star)$ for infinitely many $i$.
Additionally, for any $i$ we have $f_{i+1}(y_\star) \leq f_i(y_\star)$ and so the sequence $(f_i(y_\star))_{i=0}^\infty$ contains an infinitely decreasing sequence in $(P,<_P)$ in contradiction to well-foundedness.
\qed
\end{proof}

\begin{claim}
    \Cref{constructor:domainPointwise} is sound.
\end{claim}

\begin{proof}
Given an implicit ranking
$\rankname^\insup$
with parameters $\seq x = y\concatvar\seq z$ and soundness condition $\conditionsuper{\insup}$, 
we need to show that $\mathrm{DomPW}(\rankname^{\insup},y)$ is an implicit ranking with parameters $\seq z$ and soundness condition $\condition$ as defined in \Cref{constructor:domainPointwise}.
For a domain $\domain$, let $A^\insup, f_\insup$ be a ranking range and ranking function, respectively, for $\rankname^\insup$ and $\domain$.
Define $A = \assignset(y,\domain) \to A^\insup$, ordered by $<_{\mathrm{pw}}$, and
$f\colon\structset(\signature,\domain)\times\assignset(\seq z,\domain)\to A$ by $
f(\pairhigh )=\lambda \otherassign \in \assignset(y,\domain). f_\insup(\struct, \otherassign \concatfunc \assign)
$.

Assume
$\twopair \allowbreak \models {\conserved}(\seq x,\seq x') = \forall {y}. \conservedsuper{\insup}(y \concatvar \seq {z}, y \concatvar \seq {z'})$.
Then for any assignment $\otherassign$ to $y$ we have $\twopair,\otherassign \models \conservedsuper{\insup}(y \concatvar \seq {z}, y \concatvar \seq {z'})$, which by our conventions is equivalent to $(\struct, \otherassign\concatfunc\assign),(\struct',\otherassign\concatfunc\assign') \models \conservedsuper{\insup}(\seq {y} \concatvar \seq {z}, \seq {y} \concatvar \seq {z'})$
by the definition of $\rankname^\insup$ we get 
$ f_\insup (\struct,\otherassign \concatfunc \assign)\geq_\insup f_\insup (\struct',\otherassign \concatfunc \assign')$. 
This holds for any $\otherassign\in \assignset(y,\domain)$ so it follows that: 
$ \lambda \otherassign. f_\insup(\struct,\otherassign\concatfunc\assign)\geq_{\text{pw}}\lambda \otherassign. f_\insup(\struct',\otherassign\concatfunc\assign') $ and so 
$ f\pairhigh\geq_{\text{pw}}f\pairlow $.
If we have additionally $\twopair \models (\exists y. {\reducedsuper{\insup}}( y \concatvar \seq {z}, y \concatvar \seq {z'} ))$ then there is $\otherassign\in\assignset(y,\domain)$ such that
$ f_\insup (\struct,\otherassign \concatfunc \assign) >_\insup f_\insup (\struct',\otherassign \concatfunc \assign') $ it follows that $f\pairhigh >_{\mathrm{pw}} f\pairlow $.

For minimality, assume $ \pair \models \minformula(\seq z)  = \forall y. \minsuper{\insup}(y \concatvar \seq z)$, then assume towards contradiction that $f\pair$ is not minimal, then there exists some function $g\in A$ such that $g <_{\text{pw}} f\pair$, so there is $\otherassign \in \assignset(y,\domain)$ such that $g(\otherassign) <_\insup f\pair(\otherassign) = f_\insup(\struct,\otherassign\concatfunc \assign)$, so $f_\insup(\struct,\otherassign\concatfunc \assign)$ is not minimal in $A^\insup$. But, $\pair,\otherassign \models \minsuper{\insup}(y\concatvar\seq z)$, so from definition we get $f(\struct,\otherassign\concatfunc\assign)$ is minimal in $A^\insup$, contradiction.

Finally, consider $<_{\text{pw}} \textnormal{restricted to } P=\{ f\pair \mid \struct\in \condition, \assign\in\assignset(\seq z,\domain) \}$ for ${\condition} = \{ 
\struct \in \structset(\signature) \mid  \struct\in 
\conditionsuper{\insup} \textnormal{ and } \forall v\in \assignset(\seq z). 
|\{u\in \assignset(y) \mid (\struct,u\concatfunc v)\models \neg \minsuper{{\insup}}(y\concatvar\seq z) \}| < \infty \}$.
From the recursive case the order $<_\insup$ is well-founded on the set
$P_\insup = \{f_\insup (\struct,\otherassign\concatfunc\assign) \mid \struct\in\conditionsuper{\insup}, \otherassign\concatfunc\assign\in \assignset(y\concatvar\seq z,\domain) \} $. 
From the above $\mathrm{Fin}(\assignset(y,\domain) \to P_\insup)$ is well-founded under $<_{\text{pw}}$, we will show $P$ is contained in it.
Let $\pairhigh$ such that $\struct\in\condition$ and $\assign\in\assignset(\seq z,\domain)$, we have $\struct\in\conditionsuper{\insup}$, so for any $\otherassign\in\assignset(y,\domain)$ we have $f_\insup(\struct,\otherassign\concatfunc\assign)\in P_\insup$, it remains to show that $f\pair$ has finitely many non-minimal outputs.
Indeed, by assumption $\struct\in\condition$, so for $\assign$ we have $|\{u\in \assignset(y) \mid (\struct,u\concatfunc v)\models \neg \minsuper{{\insup}}(y\concatvar\seq z) \}| < \infty$, 
for an assignment $\otherassign\in\assignset(y,\domain)$, if $f_\insup(\struct,\otherassign\concatfunc\assign)$ is not minimal then $(\struct,\otherassign\concatfunc\assign)\nvDash \minsuper{\insup}(y\concatvar\seq z)$, 
otherwise it would have to be minimal from the recursive case. It follows that $(\struct,\otherassign\concatfunc\assign)\models \neg\minsuper{\insup}(y\concatvar\seq z)$, 
so there are a finite number of inputs to  $f\pair$ such that $f\pair(\otherassign)$ is not minimal, and so $f\pair\in \mathrm{Fin}(\assignset(y,\domain) \to P_\insup)$.
\qed
\end{proof}

\begin{claim}
Given two well-founded sets $(Y,<_Y), (P,<_P)$, $(\mathrm{Fin}(Y\to P),<_{\mathrm{lex}})$, the set of functions from $Y$ to $P$  with finitely many non-minimal images, ordered by the reverse lexicographic ordering is well-founded.
\end{claim}

\begin{proof}
Our proof imitates the proof of the Manna Dershowitz theorem~\cite{MannaDershowitz}.
Assume towards contradiction that there is an infinite sequence $(f_i)_{i=0}^\infty$ such that for all $i$, $f_i\in \mathrm{Fin}(Y\to P)$ and $f_i >_{\mathrm{lex}} f_{i+1}$.
Construct an infinite tree by the following process: add a root $r$, for the first step add vertices $(y,f_0(y))$ for every $y$ such that $f_0(y)$ is not minimal, connect them all to $r$.
In the second step, add vertices $(y,f_1(y))$ for every $y\in Y$ such that
$f_1(y) \neq f_0(y)$.
If $f_1(y) <_P f_0(y)$, connect an edge $(y,f_0(y))\to (y,f_1(y))$.
Else, $f_1(y) \nleq_P f_0(y)$, there exists $x\in Y$ such that $y <_Y x$ and $f_1(x) <_P f_0(x)$, it follows that $f_0(x)$ is not minimal and so $(x,f_0(x))$ is in the graph, connect the edge $(x,f_0(x)) \to (y,f_1(y))$.
Repeat this construction for arbitrary $i$ and $i+1$. 
The construction satisfies the following properties: after step $i$ of the construction the leaves of the graph are exactly the nodes $(y,f_i(y))$ with non-minimal $f_i(y)$; every step adds only finitely many vertices; every step adds edges only for leaf vertices; every vertex has finite degree; every step adds at least one vertex, so the graph is infinite; every edges $(x,p)\to (y,q)$ satisfies $x>_Y y$ or $x = y \land p>_P q$. 
Now, because the graph is an infinite tree with finite degrees, we utilize K\H onig's infinity lemma which states that the graph has an infinite path. 
Denote the infinite path $v_0,v_1,v_2,\ldots$. 
Notice that for every $i$ we have an edge $v_i\to v_{i+1}$, and so $v_{i+1} <_{Y\times P} v_i$ where $<_{Y\times P}$ is a the standard lexicographic ordering on $Y\times P$.
Because $<_Y$ and $<_P$ are both well-founded it follows that $<_{Y\times P}$ is also well-founded, contradiction.
\qed
\end{proof}

\begin{claim}
    \Cref{constructor:domainLexicographic} is sound.
\end{claim}

\begin{proof}
Given an implicit ranking
$\rankname^\insup$,
parameters
$\seq x =  y\concatvar\seq z$ 
soundness condition 
$\conditionsuper{\insup}$,
and a relation $\orderformula( y_1, y_2)$ over $\signature$, 
we need to show that $\mathrm{DomLex}(\rankname^{\insup},y,\orderformula)$ is an implicit ranking with parameters $\seq z$ and soundness condition $\condition$ as defined in \Cref{constructor:domainLexicographic}.
For a domain $\domain$, let $A^\insup, f_\insup$ be a ranking range and ranking function, respectively, for $\rankname^\insup$ and $\domain$.
Define $\mathrm{interp}(\orderformula)=\mathcal{P}(\assignset(y_1\cdot y_2,\domain))$ which is the set of possible interpretation of the binary relation $\orderformula$.
For a structure $\struct$, define $\interp^\orderformula(\struct)\in \mathrm{interp}(\orderformula)$ to be the interpretation of $\orderformula$ in $\struct$. 
We define the ranking range of $\rankname$ by
$A = \mathrm{interp}(\orderformula)
\times \assignset(y,\domain) \to A^\insup$,
ordered such that: 
$(\interp^\orderformula_1,a_1) \leq (\interp^\orderformula_2,a_2)$ if and only if 
$I^\orderformula:=\interp^\orderformula_1=\interp^\orderformula_2$, $\interp^\orderformula$ is a partial order, and $a_1 \leq_{\mathrm{lex}} a_2$ by the lexicographic ordering on the set of functions according to $\interp^\orderformula$.
We then define the ranking function by:
$f\colon\structset(\signature,\domain)\times\assignset(\seq z,\domain)\to A$ by 
$f(\pair)=(\interp^\orderformula(s),\lambda \otherassign \in \assignset(y,\domain). f_\insup(\struct, \otherassign \concatfunc \assign))
$.

If
$\twopair \models {\conserved}(\seq x,\seq x') = \immutorder(\orderformula)\wedge \forall y.  (\conservedsuper{\insup}(y\concatvar \seq {z},
y\concatvar \seq {z'})
\vee \exists  y_0. (
\orderformula(y,y_0)
\wedge
\reducedsuper{\insup}( {y_0}\concatvar \seq {z},
{y_0}\concatvar \seq {z'})))$.
Then, in particular we have $\interp^\orderformula(\struct)=\interp^\orderformula(\struct')$ and it is a partial order.
Then for any assignment $\otherassign$ to $y$ we have $\twopair,\otherassign \models (\conservedsuper{\insup}(y\concatvar \seq {z},
y\concatvar \seq {z'})
\vee \exists  y_0. (
\orderformula(y,y_0)
\wedge
\reducedsuper{\insup}( {y_0}\concatvar \seq {z},
{y_0}\concatvar \seq {z'})))$, 
then, either 
$ f_\insup (\struct,\otherassign \concatfunc \assign)\geq_\insup f_\insup (\struct',\otherassign \concatfunc \assign')$
or there exists $\otherassign_0\in \assignset(y,\domain)$ such that $\otherassign\concatfunc\otherassign_0\in \interp^\orderformula(\struct)$ and
$f_\insup (\struct,\otherassign_0 \concatfunc \assign) >_\insup f_\insup (\struct',\otherassign_0 \concatfunc \assign')$.
This shows that $ \lambda \otherassign \in \assignset(y,\domain). f_\insup(\struct, \otherassign \concatfunc \assign) \geq_\mathrm{lex} \lambda \otherassign \in \assignset(y,\domain). f_\insup(\struct', \otherassign \concatfunc \assign') $ according to the lexicographic ordering defined by $\interp^\orderformula(\struct)$ and 
so $f\pairhigh \geq f\pairlow$ by the defined order above.
If we have additionally $\twopair \models (\exists y. {\reducedsuper{\insup}}( y \concatvar \seq {z}, y \concatvar \seq {z'} ))$ then by the same argument there is assignment $\otherassign\in\assignset(y,\domain)$ such that
$f_\insup (\struct,\otherassign \concatfunc \assign) >_\insup f_\insup (\struct',\otherassign \concatfunc \assign')$ it follows that $ f\pairhigh > f\pairlow$.
$\minformula$ is defined as for \Cref{constructor:domainPointwise} so the proof is the same.

Finally, we need to show that $P = \{f\pair \mid \struct\in \condition, v\in \assignset(\seq z,\domain) \}$ ordered by $\leq$ defined above is well-founded.
The ranking range $A$ is $\mathrm{interp}(\orderformula)\times\assignset(y,\domain) \to A^\insup$ where pairs in $\leq$ have the same interpretation $\interp^\orderformula$ for $\orderformula$.
It is easy to see that it suffices to show that the order defined on $\assignset(y,\domain) \to A^\insup$ for a states with fixed $\interp^\orderformula$.
In other words, we need to show that 
$\{ \lambda \otherassign \in \assignset(y,\domain). f_\insup(\struct, \otherassign \concatfunc \assign) \mid \struct\in \condition \textnormal{ with } \interp^\orderformula(\struct)=\interp^\orderformula, v\in \assignset(\seq z,\domain)  \}$
ordered by $<_{\mathrm{lex}}$ for fixed $\interp^\orderformula$ is well-founded.
Define $(Y=\assignset(y,\domain),<_Y=\interp^\orderformula)$,
for a state $\struct\in\condition$, we have that $\interp^\orderformula(\struct)$ is well-founded, so $<_Y$ is well-founded. 
Take $P = \{ f_\insup(\struct,\otherassign\concatvar\assign)\mid s\in \condition,\otherassign\concatvar\assign\in\assignset(\seq x,\domain) \} $ ordered by $<_P=<_\insup$.
For a state $\struct \in \condition$ we have $\struct\in\conditionsuper{\insup}$, so the order on $P$ is well-founded by the recursive case.
Finally, for any $\struct \in \condition$ and $\assign\in\assignset(\seq z,\domain)$ we have $\mathrm{finite}_{y}(\neg \minsuper{\insup}( y\concatvar\seq z),\struct,\assign)$ and so $ \lambda \otherassign \in \assignset(y,\domain). f_\insup(\struct, \otherassign \concatfunc \assign) \in \mathrm{Fin}(Y\to P)$.
Thus, from the claim above, this set is well-founded under the $<_{\mathrm{lex}}$ order.
\qed
\end{proof}

\section{Other Constructors}
\label{appendix:constructors}
We present the other four constructors we use in our implementation that were not presented in \Cref{subsection:constructors}, these are all taken from~\cite{ImplicitRankings}, with the expected minimal formulas and soundness conditions.
\begin{constructor}
\label{constructor:binary}
The \emph{binary constructor} takes a formula $\formula(\seq x)$ over $\signature$. It returns an implicit ranking $\mathrm{Bin}(\formula)=\ranktuple$ with parameters $\seq x$ and soundness condition $\condition$ defined by:
\begin{align*}
    &{\reduced}(\seq {x},\seq {x'}) = \formula(\seq {x})
    \wedge \neg \formula'(\seq {x'}) 
    \qquad \conserved(\seq {x},\seq {x'}) = 
    \neg \formula(\seq {x})\to \neg \formula'(\seq {x'})\\
    &\minformula(\seq {x}) = \neg \formula(\seq x)
    \qquad\qquad\qquad \condition=\structset(\signature)
\end{align*}
\end{constructor}

\begin{constructor}
\label{constructor:pointwise}
The \emph{pointwise constructor} takes, for $i=1,\ldots,m$,
implicit rankings $\rankname^i=\ranktuplesup{i}$, each with parameters $\seq x$ and soundness condition $\conditionsuper{i}$.
It returns an implicit ranking $\mathrm{PW}(\rankname^1,\ldots,\rankname^m) = \ranktuple$ with parameters $\seq x$ and soundness condition $\condition$ defined by:
\begin{align*}
 &{\reduced}(\seq x,\seq x')=  \conserved(\seq x,\seq x') \wedge\bigvee_i
\reducedsuper{i}(\seq x,\seq x') 
\quad \conserved(\seq x,\seq x')=\bigwedge_i \conservedsuper{i}(\seq x,\seq x') \\
& \minformula(\seq x)=\bigwedge_{i} \minsuper{i}(\seq x)
\qquad\qquad\qquad\qquad\qquad \condition =\bigcap_i 
\conditionsuper{i}
\end{align*}
\end{constructor}

\begin{constructor}
\label{constructor:lexicographic}
The \emph{lexicographic constructor} 
takes, for $i=1,\ldots,m$,
implicit rankings $\rankname^i=\ranktuplesup{i}$, each with parameters $\seq x$ and soundness condition $\conditionsuper{i}$.
It returns an implicit ranking $\mathrm{Lex}(\rankname^1,\ldots,\rankname^m) = \ranktuple$ with parameters $\seq x$ and soundness condition $\condition$ defined by:
\begin{align*}
 &{\reduced}(\seq x,\seq x')= \bigvee_i
(\reducedsuper{i}(\seq x,\seq x') \wedge \bigwedge_{j<i} 
 \conservedsuper{j}(\seq x,\seq x'))
\quad \conserved(\seq x,\seq x')={\reduced}(\seq x,\seq x') \vee \bigwedge_i \conservedsuper{i}(\seq x,\seq x') \\
& \minformula(\seq x)=\bigwedge_{i} \minsuper{i}(\seq x)
\qquad\qquad\qquad\qquad\qquad \condition =\bigcap_i 
\conditionsuper{i}
\end{align*}
\end{constructor}

\begin{constructor}
The \emph{domain permuted-pointwise constructor} \ifshort receives \else takes \fi an
implicit ranking $\rankname^\insup=\ranktuplesup{\insup}$
with parameters $\seq x = \seq y\concatvar\seq z$ and soundness condition $\condition$, and $k\in\nat$.
It returns an implicit ranking $\mathrm{DomPerm}(\rankname^\insup,\seq y,k)=\ranktuple$ with parameters $\seq z$ and soundness condition $\condition$ defined by:
\begin{align*}
     &\reduced(\seq z,\seq z') = \tilde \exists\bijec. \ \forall \seq {y}. \ \conservedsuper{\insup}(\seq y \concatvar \seq z, \seq {y_\bijec} \concatvar \seq z' ) \wedge 
    \exists \seq {y}. \ \reducedsuper{\insup}(\seq y \concatvar \seq z, \seq {y_\bijec} \concatvar \seq z' )\\
         &\conserved(\seq z,\seq z') = \tilde \exists\bijec. \
    \forall \seq {y}. \ \conservedsuper{\insup}(\seq y \concatvar \seq z, \seq {y_\bijec} \concatvar \seq z' ) \qquad  \minformula(\seq z) = \forall \seq y. \minsuper{\insup}(\seq y \concatvar \seq z)\\
&{\condition} = \{  \struct \in \structset(\signature) \mid  \struct\in \conditionsuper{\insup} 
\textnormal{ and } \forall \assign\in\assignset(\seq z).\ \mathrm{finite}_{\seq y}(\neg \minsuper{\insup}(\seq y\concatvar\seq z),\struct,\assign) \}
\end{align*}
where:
     $\seq y_\bijec = \mathrm{ite}(\seq y = \seq {y^1_{\scriptstyle\rightarrow}},\
\seq {y^1_{\scriptstyle\leftarrow}}, \
\mathrm{ite}(\seq y = \seq {y^1_{\scriptstyle\leftarrow}},\
\seq {y^1_{\scriptstyle\rightarrow}} , \ 
\ldots,$\\
\rightline{
    $\qquad \ \ \mathrm{ite}(\seq y = \seq {y^k_{\scriptstyle\rightarrow}}, \ \seq {y^k_{\scriptstyle\leftarrow}}, \ \mathrm{ite}(\seq y = \seq {y^k_{\scriptstyle\leftarrow}}, \ \seq {y^k_{\scriptstyle\rightarrow}} , \ \seq y))))$
} 
     $\mathrm{ \tilde \exists\bijec. \ \formula } := \exists \seq {y^1_{\scriptstyle\rightarrow}},\seq {y^1_{\scriptstyle\leftarrow}},\ldots,\seq {y^k_{\scriptstyle\rightarrow}},\seq {y^k_{\scriptstyle\leftarrow}}. \bigwedge_{i<j} ({\seq y^i_{\scriptstyle\rightarrow}}\neq {\seq y^j_{\scriptstyle\rightarrow}}\wedge {\seq y^i_{\scriptstyle\leftarrow}}\neq {\seq y^j_{\scriptstyle\leftarrow}}
     \wedge
     {\seq y^i_{\scriptstyle\rightarrow}}\neq {\seq y^j_{\scriptstyle\leftarrow}}
     )\wedge \formula$ 
\noindent
\end{constructor}

\section{Example Descriptions}
\label{appendix:examples}

We provide some additional information on our evaluation examples. 
For every example, we give a short description of the system and its temporal property.

\paragraph{Paxos.} The Paxos consensus protocol~\cite{PaxosLamport} allows nodes to propose values and vote on them over rounds (ballots) in order to reach a consensus over an agreed value. 
The temporal property we verify is that under the assumptions that there is some maximal round $r_0$, after which no node proposes values\raz{?}, there is some quorum of nodes that is responsive and a value is proposed in $r_0$ then eventually a quorum of nodes agrees on that value.

\paragraph{ABP.}
The alternating bit protocol~\cite{alternating_bit_protocol} is a protocol for message passing between a sender and a receiver over lossy FIFO channels.
The sender repeatedly sends data messages along with the corresponding data and the sender's bit, the receiver repeatedly sends ack messages with the receiver's bit. 
Upon receiving a message with a different bit from its own the receiver writes down the data and flips its bit. 
Similarly, upon receiving a message with a disagreeing bit the sender flips its bit, moves on to the next data entry.
The temporal property we verify is that under weak fairness of the communication channels with unbounded message drops, all data from the sender will eventually be received by the receiver.

\paragraph{HRB-Correct \& HRB-Relay} The Hybrid Reliable Broadcast Protocol~\cite{hrb_source} allows nodes to broadcast a message under hybrid reliability guarantees.
Nodes have 4 failure modes: symmetric omission faulty, asymmetric omission faulty, symmetry faulty or arbitrary faulty. 
Under appropriate assumptions on the number of failures of every mode, nodes obedient to the protocol can only accept messages that were initiated to the network.
We verify two temporal properties.
Correctness: if a message was sent to all non-faulty nodes in the network eventually some obedient node accepts it. 
Relay: if some obedient node accepts a message then eventually all non-faulty nodes will.
Both properties are verified under assumption non-faulty nodes that received the initial message eventually send it, and non-faulty nodes that are sent a message eventually receive it.

\paragraph{Ackermann.}
The Ackermann function~\cite{ackermann1928} is a fast-growing recursive function defined by $A(0,n)=n+1, A(m,0)=A(m-1,1), A(m,n)=A(m-1,A(m,n-1))$. 
The transition system we consider is a stack implementation of the Ackermann function~\cite{iterative_ackermann}.
The temporal property we verify is termination of the transition system, which is equivalent to totality of the Ackermann function.

\paragraph{Queue Examples.}
We prove the three motivating examples from~\cite{towards_liveness_proofs}. 
In each example we prove the liveness property that any sent message is eventually received.
In TimestampedQueue, a stream of timestamped messages is passed over a FIFO queue, with a fairness assumption that the queue is polled infinitely often. 
In CascadingQueue, such a stream is passed between two consecutive queues, with the assumption that both queues are polled infinitely often.
In ReorderingQueue, two sending queues merge into a single arrival queue with the assumption that each is polled inifnitely often.

\paragraph{MutexRing.} An algorithm for mutual exclusion for nodes positioned in a ring, where access to the critical section is governed by possession of a token (from~\cite{liveness_invisible}).
Nodes can be in states: neutral, waiting or critical, may only move from neutral to critical if they hold the token,
and may only move from neutral to waiting if they do not hold the token. 
Upon exiting the critical section nodes pass the token to their right-hand neighbor.
We verify mutual exclusion under fair scheduling.

\paragraph{LeaderRing.}
The Leader Election in Ring protocol~\cite{chang_roberts}. 
A set of nodes is ordered in ring, each with a unique ID.
Every node, initially sends its ID, then waits to receive IDs from its left neighbor.
Upon receing an ID, a node passes that ID to its right neighbor if it is greater than its own ID.
A node may declare itself leader if it receives its own ID.
The temporal property we verify is that under fair scheduling of all nodes eventually a leader gets elected.

\paragraph{ToyStabilization.}
A simplified version of Dijsktra's $k$-state protocol that captures only the movement of privileges (see the motivating example in~\cite{ImplicitRankings}).  
We verify that eventually any node is scheduled. 

\paragraph{Dijsktra k-State.}
Dijkstra's $k$-state self-stabilization protocol~\cite{dijkstra_self_stab}. 
The state of the protocol is given by a ring of $n$ nodes, where each holds a value in $\{1,\ldots,k\}$ for some $k > n$.
Each node may be either privileged or not, defined according to a relation between its own value and its left-hand neighbor's value.
At each step, some privileged node takes a step with the effect of losing its privilege and potentially creating a privilege for the right-hand neighbor.
We verify that under any scheduling of privileged nodes, the system eventually moves to a stable state where exactly one node is privileged.
The proof is split into three temporal properties, based on~\cite{regular_abstractions}.

\paragraph{BinaryCounter.} A binary counter implemented in an array, decreasing from $2^n-1$ to $0$, implemented such that in every operation only one cell is modified.
We verify the program terminates.

\paragraph{SAT-Backtrack.} 
A deterministic backtracking algorithm for solving propositional SAT. 
The algorithm sets every variable, in a fixed order, first to true, and then after backtracking due to some unsatisfied clause, to false. 
If the current assignment satisfies all clauses it returns SAT, otherwise backtracks.
Upon backtracking from the first variable already set to false, it returns UNSAT.
We verify that the algorithm terminates.

\paragraph{SAT-CDCL.}
The CDCL algorithm for propositional SAT~\cite{CSP_SAT}.
The algorithm non-deterministically applies decisions on variable assignments, unit propagations, backjumping upon contradictions and learning of implied clauses.
We verify that the algorithm terminates, returning SAT or UNSAT.
\fi

\end{document}